%% file: 4prof-arxiv.tex
\documentclass{article}
\pdfoutput=1
\usepackage[margin=1in]{geometry}

\usepackage[pdftex]{graphicx}
\usepackage{amsmath}
\usepackage[psamsfonts]{amssymb}
\usepackage{amsthm}
\usepackage{verbatim}
\usepackage{algorithmic, algorithm}
\usepackage{wrapfig}
\usepackage{tikz}
\usetikzlibrary{shapes,arrows,shadows,positioning,matrix,backgrounds}
\usepackage{subcaption} 
\usepackage[T1]{fontenc}
\usepackage{url}
\title{Distributed Estimation of Graph 4-Profiles\footnote{This work will be presented in part at WWW'16, and has been supported by NSF Grants CCF 1344179, 1344364, 1407278, 1422549 and ARO YIP W911NF-14-1-0258.}}

\author{
     Ethan R. Elenberg, Karthikeyan Shanmugam, \\
     Michael Borokhovich, Alexandros G. Dimakis \\
   	The University of Texas at Austin \\
   	\texttt{\{elenberg,karthiksh,michaelbor\}@utexas.edu} \\
   	\texttt{dimakis@austin.utexas.edu} \\
}

\newcommand{\pnorm}[2]{{\Vert #1 \Vert} _{#2}}

\newtheorem{theorem}{Theorem}
\newtheorem{prop}{Proposition}
\newtheorem{defin}{Definition}
\newtheorem{lem}{Lemma}
\newtheorem{cor}{Corollary}
\newcommand{\vecalpha}{{\boldsymbol \alpha}}
\newcommand{\prof}{\textsc{4-Prof-Dist}}
\newcommand{\orca}{\textsc{Orca}}

\def\smalldist{.05}

\usepackage{hyperref}

\begin{document}

\maketitle

\begin{abstract}
	
We present a novel distributed algorithm for counting all  four-node induced subgraphs in a big graph. 
These counts, called the $4$-profile, describe a graph's connectivity properties and have found several uses ranging from bioinformatics to spam detection. We also study the more complicated problem of estimating the local $4$-profiles centered at each vertex of the graph. 
The local $4$-profile embeds 
every vertex in an $11$-dimensional space that characterizes the local geometry of its neighborhood: vertices that connect different clusters will have different local $4$-profiles compared to those that are only part of one dense cluster.

Our algorithm is a local, distributed message-passing scheme on the graph and computes all the local $4$-profiles in parallel. 
We rely on two novel theoretical contributions: we show that local $4$-profiles can be calculated using compressed two-hop information and also establish novel concentration results that show that graphs can be substantially sparsified and still retain 
good approximation quality for the global $4$-profile. 

We empirically evaluate our algorithm using a distributed GraphLab implementation that we scaled up to $640$ cores. We show that our algorithm can compute global and local $4$-profiles of graphs with millions of edges in a few minutes, significantly improving upon the previous state of the art. 

\end{abstract}

\section{Introduction}
\input{Intro}

\subsection{Related work}
\input{Relatedwork}

\section{Distributed Algorithm}\label{4prof-eq}
\input{4prof-eq}

\section{Sparsification and Concentration}\label{sec:spars}
\input{Sampling}

\section{Experiments}\label{sec:experiments}
\input{Experiments}

\bibliographystyle{IEEEtran}
\bibliography{IEEEabrv,refs}

\newpage
\appendix
\section{Appendix}
\input{Appendix}

\end{document}

%% file: Intro.tex
Graph $k$-profiles are local statistics that count the number of small subgraphs in a big graph. $k$-profiles
are a natural generalization of triangle counting and are increasingly popular for several problems in big graph analytics.
Globally, they form a concise graph description that has found several applications for the web~\cite{becchetti08,OCallaghan2012}
as well as social and biological networks~\cite{Ugander2013,przPPIorig}. Furthermore, as we explain, the \textit{local} profile of a vertex is an embedding in a low-dimensional feature space that reveals local structural information. 
Mathematically, $k$-profiles are of significant recent interest since they are connected to the emerging theory of graph homomorphisms, graph limits and graphons~\cite{Borgs2006,Ugander2013,lovasz2012large}.

There are $4$ possible graphs on $3$ vertices, labeled $H_0, \ldots, H_3$, as in Figure \ref{fig:profiles} (left).
The (global) $3$-profile of a graph $G(V,E)$ is a vector having one coordinate for each distinct $H_i$ that counts how many times 
that $H_i$ appears as an induced subgraph of $G$. For example, the graph $G=K_4$ (the complete graph on $4$ vertices) has the $3$-profile $[0,0,0,4]$ since it contains $4$ triangles and no other (induced) subgraphs. The graph $C_5$ (the cycle on $5$ vertices, \textit{i.e.} a pentagon) has the $3$-profile $[0,5,5,0]$. Note that the sum of the $k$-profile is always $\binom{|V|}{k}$, the total number of subgraphs.
Estimating $3$-profiles of big graphs is a topic that has received attention from several communities recently (\textit{e.g.}, see \cite{Ugander2013,Williams2014mod,Hocevar2014bio,ourpaperKDD} and references therein). 

In this paper we are interested in the significantly more challenging problem of estimating $4$-profiles. Figure \ref{fig:profiles} (right)
shows the $11$ possible graphs on $4$ vertices,\footnote{Actually there are $17$ local subgraphs when  considering vertex automorphisms. This is discussed in Section \ref{4prof-eq} in detail. For the purpose of initial exposition, we will ignore vertex automorphisms.} labeled as $F_i$, $\, i=0 \ldots 10$. Given a big graph $G(V,E)$ we are interested in estimating the global $4$-profile, \textit{i.e.} count how many times each $F_i$ appears as an induced subgraph of $G$. In addition to global graph statistics, 
we are interested in local 4-profiles: given a specific vertex $v_0$,  the local $4$-profile of $v_0$ is an 11-dimensional vector, with each coordinate $i$ counting how many induced $F_i$'s contain $v_0$. In Figure \ref{fg:ex2} we show an example of the local $4$-profile of a vertex. 

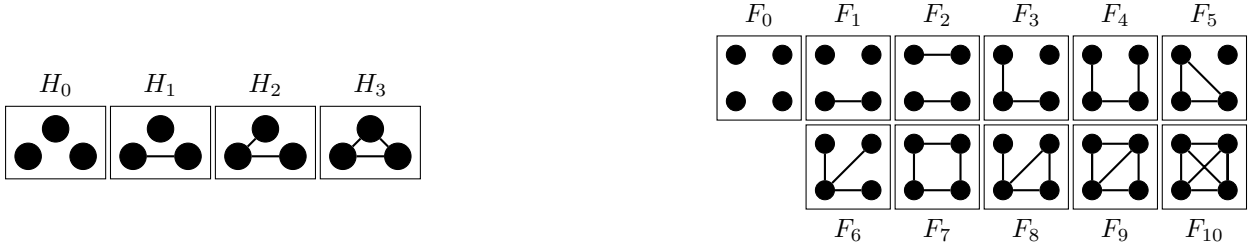
\begin{figure}[ht]
 \begin{subfigure}{ .2 \textwidth}	 
    \begin{tikzpicture}[
      inner/.style={circle,draw,fill=black,inner sep=1pt, minimum size=1em},
      outer/.style= {draw
     }
    ]
    \matrix (w0) [matrix of nodes, outer, nodes={inner}, label=$H_0$]{
      &{} \\
      {} & & {}\\
    };
    \matrix (w1) [matrix of nodes, outer, nodes={inner},right=\smalldist of w0, label=$H_1$]{
        &{} \\
       {} & & {}\\
     };
     \matrix (w2) [matrix of nodes, outer, nodes={inner},right=\smalldist of w1, label=$H_2$]{
        &{} \\
       {} & & {}\\
     };
     \matrix (w3) [matrix of nodes, outer, nodes={inner},right=\smalldist of w2, label=$H_3$]{
        &{} \\
       {} & & {}\\
     };
    \draw[thick] (w1-2-1)--(w1-2-3);
     \draw[thick] (w2-2-3)--(w2-2-1)--(w2-1-2);
     \draw[thick] (w3-2-1)--(w3-1-2)--(w3-2-3)--(w3-2-1);
    \end{tikzpicture}
 \end{subfigure}
 \hspace{6cm}
\begin{subfigure}{.5 \textwidth}
 	\begin{tikzpicture}[
 	inner/.style={circle,draw,fill=black,inner sep=1pt, minimum size=.72em},
 	outer/.style= {draw
 	}
 	]
	\vspace{-1cm}
 	\matrix (w0) [matrix of nodes, outer, nodes={inner}, label=$F_0$, column sep=10pt, row sep=10pt]{
 	     {} &{} \\ 
 	    {} & {}\\
 	  };
 	  \matrix (w1) [matrix of nodes, outer, nodes={inner},right=\smalldist of w0, label=$F_1$, column sep=10pt, row sep=10pt]{
 	      {} &{} \\ 
 	         {} & {}\\
 	    };
 	  \matrix (w2) [matrix of nodes, outer, nodes={inner},right=\smalldist of w1, label=$F_2$, column sep=10pt, row sep=10pt]{
 	      {} &{} \\ 
 	         {} & {}\\
 	    };
 	  \matrix (w3) [matrix of nodes, outer, nodes={inner},right=\smalldist of w2, label=$F_3$, column sep=10pt, row sep=10pt]{
 	    {} &{} \\ 
 	       {} & {}\\
 	    };
 		\matrix (w4) [matrix of nodes, outer, nodes={inner},right=\smalldist of w3, label=$F_4$, column sep=10pt, row sep=10pt]{
 		{} &{} \\ 
 		{} & {}\\
 		};
 		\matrix (w5) [matrix of nodes, outer, nodes={inner},right=\smalldist of w4, label=$F_5$, column sep=10pt, row sep=10pt]{
 		{} &{} \\ 
 		{} & {}\\
 		};
 		\matrix (w6) [matrix of nodes, outer, nodes={inner},below=\smalldist of w1, label=below:$F_6$, column sep=10pt, row sep=10pt]{
 		{} &{} \\ 
 		{} & {}\\
 		};
 		\matrix (w7) [matrix of nodes, outer, nodes={inner},right=\smalldist of w6, label=below:$F_7$, column sep=10pt, row sep=10pt]{
 		{} &{} \\ 
 		{} & {}\\
 		};
 		\matrix (w8) [matrix of nodes, outer, nodes={inner},right=\smalldist of w7, label=below:$F_8$, column sep=10pt, row sep=10pt]{
 		{} &{} \\ 
 		{} & {}\\
 		};
 		\matrix (w9) [matrix of nodes, outer, nodes={inner},right=\smalldist of w8, label=below:$F_9$, column sep=10pt, row sep=10pt]{
 		{} &{} \\ 
 		{} & {}\\
 		};
 		\matrix (w10) [matrix of nodes, outer, nodes={inner},right=\smalldist of w9, label=below:$F_{10}$, column sep=10pt, row sep=10pt]{
 		{} &{} \\ 
 		{} & {}\\
 		};
 	\draw[thick] (w1-2-1)--(w1-2-2);
 	\draw[thick] (w2-2-2)--(w2-2-1);
 	\draw[thick] (w2-1-2)--(w2-1-1);
 	\draw[thick] (w3-1-1)--(w3-2-1)--(w3-2-2);
 	\draw[thick] (w4-1-1)--(w4-2-1)--(w4-2-2)--(w4-1-2);
 	\draw[thick] (w5-1-1)--(w5-2-1)--(w5-2-2)--(w5-1-1);
 	\draw[thick] (w6-1-1)--(w6-2-1)--(w6-2-2);
 	\draw[thick] (w6-1-2)--(w6-2-1);
 	\draw[thick] (w7-1-1)--(w7-2-1)--(w7-2-2)--(w7-1-2)--(w7-1-1);
 	\draw[thick] (w8-1-1)--(w8-2-1)--(w8-2-2)--(w8-1-2)--(w8-2-1);
 	\draw[thick] (w9-1-1)--(w9-2-1)--(w9-2-2)--(w9-1-2)--(w9-2-1)--(w9-1-1)--(w9-1-2);
 	\draw[thick] (w10-1-1)--(w10-2-1)--(w10-2-2)--(w10-1-2)--(w10-1-1)--(w10-2-2)--(w10-1-2)--(w10-2-1);
  \end{tikzpicture}
\end{subfigure}
\vspace{-2mm}	
\caption{Left: The $4$ possible non-isomorphic graphs on $3$ vertices used to calculate the $3$-profile of a graph $G$. The 3-profile counts how many times each $H_i$ appears in $G$. Right: The $11$ non-isomorphic graphs on $4$ vertices used to calculate the $4$-profile of a graph. }
 \label{fig:profiles}
 \vspace{-2mm}
 \end{figure}

\begin{figure}[ht]
\centering
	\begin{tikzpicture}[
	inner/.style={circle,draw,fill=white,inner sep=1pt, minimum size=1em},
	outer/.style= {draw}
	]
	\matrix (w0) [matrix of nodes, nodes={inner}, column sep=15pt, row sep=15pt]{
		 & {$v_1$} & \\
	    {$v_0$} &{$v_2$} & {$v_4$} \\
	     & {$v_3$} & \\
	  };
	\draw[thick] (w0-2-1)--(w0-1-2)--(w0-2-3)--(w0-3-2)--(w0-2-1)--(w0-2-2)--(w0-1-2);
\end{tikzpicture}
\vspace{-2mm}
\caption{An example for local profiles. The global $3$-profile is $[0,3,6,1]$. The global $4$-profile is $[0,0,0,0,2,0,0,1,2,0,0]$. The local $4$-profile of $v_0$ is $[0,0,0,0,1,0,0,1,2,0,0]$. The first $1$ in the profile corresponds to the subgraph $F_4$. Notice that $v_0$ participates 
in only one $F_4$, jointly with vertices $v_2,v_3,v_4$. 
}
\label{fg:ex2}
\vspace{-2mm}
\end{figure}
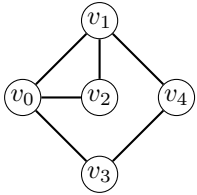

The local $4$-profile of a vertex can be seen as an embedding in an $11$-dimensional space that characterizes the local geometry of its neighborhood: vertices that connect different clusters will have different local $4$-profiles compared to those that are only part of one dense cluster. A very naive estimation of $4$-profiles requires examining $\binom{n}{4}$ possible subgraphs. Furthermore, for estimating each local $4$-profile independently, this computation has to be repeated $n$ times, once for each vertex. 
Note that the local $4$-profiles may be rescaled and added together to obtain the global $4$-profile. Since some of the $4$-profile subgraphs are disconnected (like $F_0, F_1, F_5$), local $4$-profiles contain information beyond the local neighborhood of a vertex. Therefore, in a distributed setting, it seems that global communication is required.

\subsection{Our contributions}
\vspace{-2mm}
Surprisingly, we show that very limited global information is sufficient to calculate all local $4$-profiles and that it can be re-used to calculate all the local $4$-profiles in parallel. Specifically, we introduce a distributed algorithm to estimate all the local $4$-profiles and the global profile of a big graph. This restrictive setting does not allow communication between nonadjacent vertices, a key component of previous centralized, shared-memory approaches. Our algorithm relies on two novel theoretical results: \\
\textbf{Two-hop histograms are sufficient:} Our algorithm operates by having each vertex first perform local message-passing to its neighbors and then solve a novel system of equations for the local $4$-profile.
Focusing on a vertex $v_0$, the first easy step is to calculate its local $3$-profile. It can be shown that the local $3$-profile combined with the full two-hop connectivity information is sufficient to estimate the local $4$-profile for each vertex $v_0$. This is not immediately obvious, since naively counting the $3$-path (an automorphism of $F_4$) would require $3$-hop connectivity information.
However, we show that less information needs to be communicated. Specifically, we prove that the triangle list combined with what we call the \textit{two-hop histogram} is sufficient: for each vertex $v_i$ that is 2-hops from $v_0$, we only need 
the \textit{number of distinct paths} connecting it to $v_0$, not the full two hop neighborhood. If the two-hop neighborhood is a tree, this amounts to no compression. However, for real graphs the two-hop histogram saves a factor of $3$x to $5$x in communication in our experiments. This enables (Section 4) an even more significant running time speedup of $5-10$ times on several distributed experiments using $12-20$ compute nodes. \\
\textbf{Profile Sparsification:} One idea that originated from triangle counting~\cite{TsourakakisSubsamp2009,Tsourakakis2011sparsifier} 
is to first perform random subsampling of edges to create a sparse graph called a \textit{triangle sparsifier}. Then count the number of triangles in the sparse graph and re-scale appropriately to estimate the number in the original graph. The main challenge is proving that the randomly sparsified graph has a number of triangles sufficiently concentrated around its expectation. Recently this idea was generalized to $3$-profile sparsifiers in \cite{ourpaperKDD}, with concentration results for estimating the full $3$-profile. These papers rely on Kim-Vu polynomial concentration techniques~\cite{KimVu2000concentration} that scale well in theory, but typically the estimated errors are orders of magnitude larger than the measured quantities for reasonable graph sizes.
In this paper, we introduce novel concentration bounds for global $k$-profile sparsifiers that use a novel information theoretic 
technique called read-$k$ functions~\cite{GavinskyReadK}. Our read-$k$ bounds allow usable concentration inequalities for sparsification factors of approximately $0.4$ or higher (Section \ref{sec:results}). Note that removing half the edges of the graph does not accelerate the running time by a factor of $2$, but rather by a factor of nearly $8$, as shown in our experiments. \\
\textbf{System implementation and evaluation:} We implemented our algorithm using GraphLab PowerGraph~\cite{powergraphGAS2012} and tested it in multicore and distributed systems scaling up to $640$ cores. The benefits of two-hop histogram compression and sparsification allowed us to compute the global
and local $4$-profiles of very large graphs. For example, for a graph with $5$ million vertices and $40$ million edges we estimated the global $4$-profile in less than $10$ seconds. For computing all local $4$-profiles on this graph, the previous state of the art~\cite{Hocevar2014bio} required $1200$ seconds while our algorithm required less than $100$ seconds.

%% file: Relatedwork.tex
\vspace{-2mm}
The problem of counting triangles in a graph has been addressed in distributed \cite{Schank2007} and streaming \cite{becchetti08} settings, and this is a standard analytics task for graph engines \cite{Satish2014}. The Doulion algorithm \cite{TsourakakisSubsamp2009} estimates a graph's triangle count via simple edge subsampling. Other recent work analyzes more complex sampling schemes 
\cite{seshadri2012wedge,seshadhriSublinearFOCS2015} 
and extends to approximately counting certain $4$-subgraphs \cite{Ahmed2014,Jha2014}.
 Mapreduce algorithms for clique counting were introduced by Finocchi et al. \cite{Finocchi2014}. Our approach is similar to that of \cite{ourpaperKDD}, which calculates all $3$-subgraphs and a subset of $4$-subgraphs distributedly using \textit{edge pivots}. In this work we introduce the $2$-hop histogram to compute all $4$-subgraphs.

Concentration inequalities for the number of triangles in a random graph have been studied extensively. The standard method of martingale bounded differences (McDiarmid's inequality) is known to yield weak concentrations around the mean for this problem. The breakthrough work of Kim and Vu \cite{KimVu2000concentration} provides superior asymptotic bounds by analyzing the concentration of multivariate polynomials.
This was later improved and generalized in \cite{Janson2002}, and 
applied to subsampled triangle counting in \cite{Tsourakakis2011sparsifier}. Our analysis uses a different technique called read-$k$ functions \cite{GavinskyReadK} that produces sharper concentration results for practical problem sizes.\footnote{Even though concentrations using Kim-Vu become tighter asymptotically, this happens for graphs with well over $10^{13}$ edges (see also Figure \ref{fig:accuracy1ACC}).}

Previous systems of equations relating clique counts to other $4$-subgraphs appear in \cite{Kowaluk2013eqs}, \cite{Williams2014mod}, \cite{Hocevar2014bio}, and \cite{Ahmed2015ICDM}. However, these are applied in a centralized setting and depend on information collected from nonadjacent vertices. In this work, we use additional equations to solve the same system by sharing only local information over adjacent vertices.
The connected $4$-subgraphs, or graphlets \cite{przPPIorig}, have found applications in fields such as bioinformatics 
\cite{Shervashidze2009Learning}
and computational neuroscience \cite{Fei2014graphletMRI}. In \cite{LavAppFullProfiles}, authors use \emph{all} global 4-subgraphs to analyze neuronal networks. We evaluate our algorithm against $\orca$ \cite{Hocevar2014bio}, a centralized $4$-graphlet counting algorithm, as well as its GPU implementation \cite{MilinkovicGPU2015}. Notice that while $\orca$ calculates only connected 4-subgraphs, our algorithm calculates all the connected and the disconnected 4-subgraphs for each vertex.

Concurrent with the writing of this paper, a parallel algorithm for $4$-subgraph counting was introduced in \cite{Ahmed2015ICDM}. 
Our algorithm differs by working within GraphLab PowerGraph's Gather-Apply-Scatter framework instead of the native, multithreaded C\verb!++! implementation of~\cite{Ahmed2015ICDM}. 
In terms of empirical performance, both our work and~\cite{Ahmed2015ICDM} show similar running time improvements of one order of magnitude over \orca. A more detailed comparison would depend on the hardware and datasets used. 
More importantly, our work focuses on a distributed (as opposed to multicore parallel) framework, and for our setting minimizing communication is critical.

Our theoretical results are significantly different from~\cite{Ahmed2015ICDM} and may be useful in improving that system also. Specifically, \cite{Ahmed2015ICDM} explicitly counts the number of $4$-cycles ($F_7$ in Figure \ref{fig:profiles}, Right) whereas our results show that it is possible to use only two-hop histograms instead. This results in less communication overhead, but this benefit is perhaps not as significant for shared-memory multicore platforms. 
Our second theoretical result, the novel sparsification concentration bounds, can be used for any subgraph estimation algorithm and quantify a provable tradeoff between speed and accuracy.

%% file: 4prof-eq.tex
\begin{figure}[t]
   \centering
   \begin{subfigure}{.45 \textwidth}
   \vspace{8em}
   \begin{tikzpicture}[
   inner/.style={circle,draw,fill=black,inner sep=1pt, minimum size=.72em},
   outer/.style= {draw
   }
   ]
     \matrix (w0) [matrix of nodes, outer, nodes={inner}, label=$H_0(v)$]{
            &{} \\
           |[fill=white]| {} & & {}\\
         };
       \matrix (w1) [matrix of nodes, outer, nodes={inner}, right=\smalldist of w0, label=$H_1^{e}(v)$]{
          & {} \\
         |[fill=white]| {} & & {}\\
       };
     \matrix (w1d) [matrix of nodes, outer, nodes={inner},right=\smalldist of w1, label=$H_1^{d}(v)$]{
            & |[fill=white]|{} \\
            {} & & {}\\
         };
      \matrix (w2) [matrix of nodes, outer, nodes={inner},right=\smalldist of w1d, label=$H_2^c(v)$]{
          &{} \\
         |[fill=white]|{} & & {}\\
       };
       \matrix (w2e) [matrix of nodes, outer, nodes={inner},right=\smalldist of w2, label=$H_2^{e}(v)$]{
           & |[fill=white]|{} \\ 
              {} & & {}\\
           };
     \matrix (w3) [matrix of nodes, outer, nodes={inner},right=\smalldist of w2e, label=$H_3(v)$]{
          &{} \\
             |[fill=white]| {} & & {}\\
       };
   \draw[thick] (w1-2-1)--(w1-2-3);
   \draw[thick] (w1d-2-1)--(w1d-2-3);
     \draw[thick] (w2-2-3)--(w2-2-1)--(w2-1-2);
       \draw[thick] (w2e-2-3)--(w2e-2-1)--(w2e-1-2);
     \draw[thick] (w3-2-1)--(w3-1-2)--(w3-2-3)--(w3-2-1);
   \end{tikzpicture}
   \caption{}
   \end{subfigure}
   \begin{subfigure}{ .45 \textwidth}
   \begin{tikzpicture}[
   inner/.style={circle,draw,fill=black,inner sep=1pt, minimum size=.72em},
   outer/.style= {draw
   }
   ]
   \matrix (w3a) [matrix of nodes, outer, nodes={inner}, label=$F_3(v)$, label=left:$F$, column sep=10pt, row sep=10pt]{
        |[fill=white]|{} &{} \\ 
       {} & {}\\
     };
   \matrix (w4a) [matrix of nodes, outer, nodes={inner}, right=\smalldist of w3a, label=$F_4(v)$, column sep=10pt, row sep=10pt]{
        |[fill=white]|{} &{} \\ 
       {} & {}\\
     };
     \matrix (w6a) [matrix of nodes, outer, nodes={inner},right=\smalldist of w4a, label=$F_6(v)$, column sep=10pt, row sep=10pt]{
         |[fill=white]|{} &{} \\ 
            {} & {}\\
       };
     \matrix (w8a) [matrix of nodes, outer, nodes={inner},right=\smalldist of w6a, label=$F_8(v)$, column sep=10pt, row sep=10pt]{
         |[fill=white]|{} &{} \\ 
            {} & {}\\
       };
     \matrix (w9a) [matrix of nodes, outer, nodes={inner},right=\smalldist of w8a, label=$F_9(v)$, column sep=10pt, row sep=10pt]{
         |[fill=white]|{} &{} \\ 
            {} & {}\\
       };
     \matrix (w3b) [matrix of nodes, outer, nodes={inner},below=\smalldist of w3a, label=left:$F^{'}$, column sep=10pt, row sep=10pt]{
       {} &{} \\ 
          |[fill=white]|{} & {}\\
       };
     \matrix (w4b) [matrix of nodes, outer, nodes={inner},below=\smalldist of w4a, column sep=10pt, row sep=10pt]{
       {} &{} \\ 
          |[fill=white]|{} & {}\\
       };
   \matrix (w6b) [matrix of nodes, outer, nodes={inner},below=\smalldist of w6a, column sep=10pt, row sep=10pt]{
       {} &{} \\ 
          |[fill=white]|{} & {}\\
       };
     \matrix (w8b) [matrix of nodes, outer, nodes={inner},below=\smalldist of w8a, column sep=10pt, row sep=10pt]{
       {} &{} \\ 
          |[fill=white]|{} & {}\\
       };
     \matrix (w8c) [matrix of nodes, outer, nodes={inner},below=\smalldist of w8b, label={[xshift = -650*\smalldist]180:$F^{''}$}, column sep=10pt, row sep=10pt]{ 
       {} &|[fill=white]|{} \\ 
          {} & {}\\
       };
     \matrix (w9b) [matrix of nodes, outer, nodes={inner},below=\smalldist of w9a, column sep=10pt, row sep=10pt]{
       {} &{} \\ 
          |[fill=white]|{} & {}\\
       };
   \draw[thick] (w3a-1-1)--(w3a-2-1)--(w3a-2-2);
   \draw[thick] (w3b-1-1)--(w3b-2-1)--(w3b-2-2);
   \draw[thick] (w4a-1-1)--(w4a-2-1)--(w4a-2-2)--(w4a-1-2);
   \draw[thick] (w4b-1-1)--(w4b-2-1)--(w4b-2-2)--(w4b-1-2);
   \draw[thick] (w6a-1-1)--(w6a-2-1)--(w6a-2-2)--(w6a-2-1)--(w6a-1-2);
   \draw[thick] (w6b-1-1)--(w6b-2-1)--(w6b-2-2)--(w6b-2-1)--(w6b-1-2);
   \draw[thick] (w8a-1-1)--(w8a-2-1)--(w8a-2-2)--(w8a-1-2)--(w8a-2-1);
   \draw[thick] (w8b-1-1)--(w8b-2-1)--(w8b-2-2)--(w8b-1-2)--(w8b-2-1);
   \draw[thick] (w8c-1-1)--(w8c-2-1)--(w8c-2-2)--(w8c-1-2)--(w8c-2-1);
   \draw[thick] (w9a-1-1)--(w9a-2-1)--(w9a-2-2)--(w9a-1-2)--(w9a-2-1)--(w9a-1-1)--(w9a-1-2);
   \draw[thick] (w9b-1-1)--(w9b-2-1)--(w9b-2-2)--(w9b-1-2)--(w9b-2-1)--(w9b-1-1)--(w9b-1-2);
   \end{tikzpicture} 
       \caption{}
   \end{subfigure}
    \caption{Unique (a) $3$-subgraphs and (b) $4$-subgraphs from the perspective of the white vertex is $v$. $F_8$ is the only subgraph with a third vertex automorphism $F_8^{''}$ because no other subgraph contains vertices with $3$ different degrees.}
    \label{fig:awsc3ego}
\end{figure}
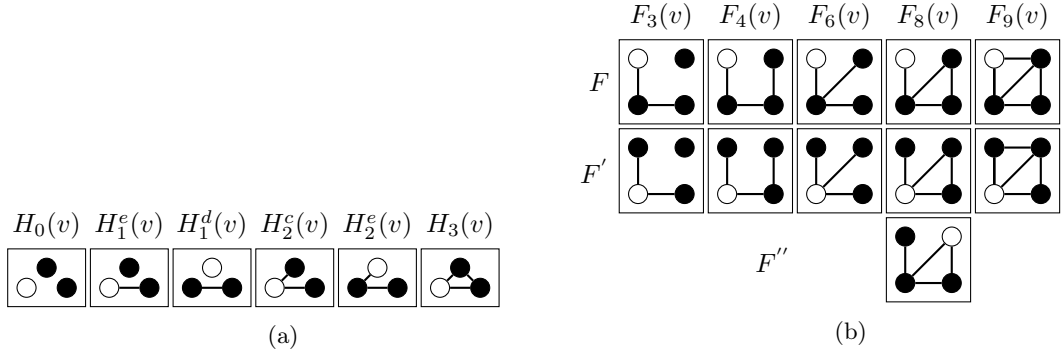

In this section, we describe $\prof$, our algorithm for computing the exact $4$-profiles in a distributed manner.
To the best of our knowledge, this is the first \textit{distributed} algorithm for calculating 4-profiles.
The key insight is to cast existing and novel equations into the GraphLab PowerGraph framework \cite{powergraphGAS2012} to get implicit connectivity information about vertices outside the $1$-hop neighborhood. Specifically, we construct the local $4$-profile from local $3$-profile, local $4$-clique count, and additional histogram information which describes the number of paths to all $2$-hop neighbors.

\begin{theorem}
There is a distributed algorithm that computes the exact local $4$-profile of a graph given each vertex has stored its local $3$-profile, triangle list, and $2$-hop histogram.
\end{theorem}

Note that the local $4$-profiles at each vertex can be added and appropriately rescaled (using the symmetries of each subgraph, also called automorphism orbits~\cite{przPPIorig}) to obtain the global $4$-profile. 

$\prof$ is implemented in the Gather-Apply-Scatter (GAS) framework \cite{powergraphGAS2012}. A distributed algorithm in this framework has 3 main phases: Gather, Apply and Scatter. Every vertex and edge has stored data which is acted upon. During the Gather phase, a vertex can access all its adjacent edges and neighbors and \emph{gather} data they possess, \textit{e.g.}, neighbor ID, using a custom reduce operation $\oplus$ (\textit{e.g.}, addition, concatenation). The accumulated information is available for a vertex at the next phase, Apply, in which it can change its own data. In the final Scatter phase, every edge sees the data of its (incident) vertices and operates on it to modify the edge data.  All nodes start each phase simultaneously, and if needed, the whole GAS cycle is repeated until the algorithm's completion.

$\prof$ solves a slightly larger problem of keeping track of counts of $17$ unique subgraphs up to vertex automorphism (see Figure \ref{fig:awsc3ego}). 
We will describe a full rank system of equations which is sufficient to calculate the local $4$-profile at every $v \in V$. 
The following subsections each explain a component of $\prof$. These separate routines are combined efficiently in Algorithm \ref{alg:4local} to calculate the local $4$-profile in a small number of GAS cycles.

\subsection{Edge pivot equations}

\begin{figure}[t]
	\centering
	\begin{tikzpicture}[
	inner/.style={circle,draw,fill=black,inner sep=.75pt, minimum size=.75em},
	outer/.style= {
	},
	marked/.style={
		fill=cyan,
	}
	]
	\matrix (c2) [matrix of nodes, outer, nodes={inner}, label={[label distance=.1cm]below:$\sum_{a \in \Gamma(v)} n_{2,va}^c n_{2,va}^e$}, column sep=10pt, row sep=10pt, ampersand replacement=\&]{
		\& {} \& \\
		|[fill=white]|{$v$} \& \& |[fill=brown!25]|{$a$} \\ 
		\& {} \& \\
	};
	\matrix (w9) [matrix of nodes, outer, nodes={inner}, right=20*\smalldist of c2, label={[label distance=.41 cm]below:$F_4^{'}(v)$}, column sep=10pt, row sep=10pt, ampersand replacement=\&]{
			{} \& {} \\ 
		|[fill=white]|{} \& |[fill=brown!25]|{}\\
	};
	\matrix (w10) [matrix of nodes, outer, nodes={inner},right=9*\smalldist of w9, label={[label distance=.5 cm]below:$2F_{7}(v)$}, column sep=10pt, row sep=10pt, ampersand replacement=\&]{
			{} \& {} \\ 
		|[fill=white]|{} \& |[fill=brown!25]|{}\\
	};
	\node[right=6*\smalldist of c2] (eq){$=$};
	\node[right=.5*\smalldist of w9] (pl){$+$};
	\draw[thick] (c2-1-2)--(c2-2-1)--(c2-2-3)--(c2-2-3)--(c2-3-2);
	\draw[thick] (w9-1-1)--(w9-2-1)--(w9-2-2)--(w9-1-2);
	\draw[thick] (w10-1-1)--(w10-2-1)--(w10-2-2)--(w10-1-2)--(w10-1-1);
	\end{tikzpicture}
	\caption{Edge pivot equation for vertex $v$ counting triangles as edges $va$ pivot about their common vertex $v$. The subgraphs $F_4^{'}(v)$ and $F_{7}(v)$ differ by one edge.}
	\label{fig:pivot}
\end{figure}

\vspace{-2mm}
The majority of our equations relate the local $4$-profile to neighboring local $3$-profiles with \textit{edge pivots} \cite{ourpaperKDD}. 
At each vertex $v$, each combinatorial equation relates a linear combination of the local $4$-subgraph counts to the count of a pair of $3$-subgraphs sharing an edge $va$. Some of these equations appear in a centralized setting in previous literature ({\cite{Kowaluk2013eqs}, \cite{Williams2014mod}, \cite{Hocevar2014bio}, \cite{Ahmed2015ICDM}).
In our algorithm, the $3$-subgraph pair count accumulates at $v$ as all incident edges $va$ \textit{pivot} over it. The edges fixed by a specific $3$-subgraph pair correspond to common edges among a subset of $4$-subgraphs.
Before that, in an initial GAS round, each vertex $v$ must gather the ID of each vertex in its neighborhood, \textit{i.e.} $a \in \Gamma(v)$, and the following quantities must be stored at each edge $va$ during Scatter phase:
\begin{align}
\begin{split}
n_{1,va}^e &= |\overline{\Gamma(v) \cup \Gamma(a)}| = |V| - (|\Gamma(v)| + |\Gamma(a)| - |\Gamma(v)\cap \Gamma(a)|), \\
n_{2,va}^c &= |\Gamma(v) \backslash \{\Gamma(a) \cup a\}| = |\Gamma(v)| - |\Gamma(v) \cap \Gamma(a)| - 1, \\
n_{2,va}^e &= n_{2,av}^c, \\
n_{3,va} &= |\Gamma(v) \cap \Gamma(a)| . \label{eq:local3}
\end{split}
\end{align}
\noindent\textbf{Gather:} The above quantities are summed at each vertex $v$ to calculate the local $3$-profile at $v$. For example, $n_{3,v} = \frac{1}{2}\sum_{a}n_{3,va}$. In addition, we gather the sum of functions of pairs of these quantities forming $13$ \textit{edge pivot} equations.
\begin{align}
\begin{split}
\sum_{a \in \Gamma(v)} \binom{n_{1,va}^e}{2} &= F_1(v) + F_2(v) , \\
\sum_{a \in \Gamma(v)} \binom{n_{2,va}^c}{2}  &= 3F_6^{'}(v) + F_8^{'}(v),  \\
\sum_{a \in \Gamma(v)} \binom{n_{3,va}}{2} &= F_9^{'}(v) + 3F_{10}(v) , \\
\sum_{a \in \Gamma(v)} n_{1,va}^en_{2,va}^c &= 2F_3^{'}(v) + F_4^{'}(v), \\
\sum_{a \in \Gamma(v)} n_{1,va}^en_{3,va} &= 2F_5(v) + F_8^{''}(v) , \\
\sum_{a \in \Gamma(v)} n_{2,va}^c n_{2,va}^e &= F_4^{'}(v) + 2F_7(v), \\
\sum_{a \in \Gamma(v)} n_{2,va}^c n_{3,va} &= 2F_8^{'}(v) + 2F_9^{'}(v) , \\
n_{1,v}^d |\Gamma(v)| &= F_2(v) + F_4(v) + F_8(v). \label{eq:global4pivot} 
\end{split}
\end{align}
The primed notation differentiates between subgraphs of different automorphism orbits, as in Figure \ref{fig:awsc3ego}. By accumulating pairs of $3$-profile structures as in \eqref{eq:global4pivot}, we receive aggregate connectivity information about vertices more than $1$ hop away. Consider the sixth equation as an example. The product between $n_{2,va}$ and $n_{2,va}^e$ subgraphs along edge $va$ forms $4$-node graphs with the following structural constraints: three vertex pairs are connected, two vertex pairs are disjoint, and one pair may be either connected or disjoint. $F_4^{'}(v)$ and $F_{7}(v)$ satisfy these constraints and differ on the unconstrained edge. Thus, as shown in Figure \ref{fig:pivot}, they both contribute to the sum of $n_{2,va}^c n_{2,va}^e$.

The following edge pivot equations are linearly independent when solving for the local $4$-profile only. Note the last $2$ equations require calculating the local $3$-profile:
\begin{align}
\begin{split}
\sum_{a \in \Gamma(v)} \binom{n_{2,va}^e}{2} &= F_6(v) + F_8(v),  \\
\sum_{a \in \Gamma(v)} n_{1,va}^en_{2,va}^e &= F_3(v) + F_4(v), \\
\sum_{a \in \Gamma(v)} n_{2,va}^e n_{3,va} &= F_8^{''}(v) + 2F_9(v) , \\
\sum_{a \in \Gamma(v)} n_{3,a} - n_{3,va} &= F_8(v) + 2F_9(v) + 3F_{10}(v) , \\
\sum_{a \in \Gamma(v)} n_{2,a}^e - n_{2,va}^c &= F_4(v) + 2F_7(v) + F_8^{''}(v) + 2F_{9}^{'}(v). \label{eq:local3end}
\end{split}
\end{align}

\noindent\textbf{Apply:} The left hand sides of all such equations are stored at $v$.

\subsection{Clique counting}
\vspace{-2mm}
The aim of this subtask is to count $4$-cliques that contain the vertex $v$. For this, we accumulate a list of triangles at each vertex $v$. Then, at the Scatter stage for every $va$, it is possible to check if neighbors common to $v$ and $a$ have an edge between them. This implies a $4$-clique. 

\noindent\textbf{Scatter:} In addition to the intersection size $|\Gamma(v) \cap \Gamma(a)|$ at each edge $va$ as before, we now require the intersection list $\{b : b \in \Gamma(v), b \in \Gamma(a) \}$ as a starting point.

\noindent\textbf{Gather, Apply:} The intersection list is gathered at each vertex $v$. This produces all pairs of neighbors in $\Gamma(v)$ which are adjacent, \textit{i.e.} all triangles containing $v$. It is stored as $\Delta(v)$ during the Apply stage at $v$.

\noindent\textbf{Gather, Apply:} 
Each edge $va$ computes the number of $4$-cliques by counting how many pairs in $\Delta(a)$ contain exactly two neighbors of $v$. We use a similar equation to calculate $F_8(v)$ concurrently:
\begin{align}
\begin{split}
\sum_{a \in \Gamma(v)}  |(b,c) \in \Delta(a) : b \in \Gamma(v), c \in \Gamma(v)| &= 3F_{10}(v) ,\\ 
\sum_{a \in \Gamma(v)}  |(b,c) \in \Delta(a) : b \notin \Gamma(v), c \notin \Gamma(v)| &= F_{8}(v) . \label{eq:h8direct}
\end{split}
\end{align}
At the Apply stage, store the left hand sides as vertex data. 

\subsection{Histogram 2-hop information}\label{subsec:nh}
\vspace{-2mm}
Instead of calculating the number of cycles $F_7(v)$ directly, we can simply construct another linearly independent equation and add it to our system. 
Let each vertex $a$ have a vector of (vertex ID, count) pairs $(p,c_a[p])$ for each of its adjacent vertices $p$. Initially, $c_a[p] = 1$ and this \textit{histogram} contains the same information as $\Gamma(a)$.
For any $a \in \Gamma(v)$ and $p \notin \Gamma(v)$, $c_a[p] = 1 \Leftrightarrow $ $vap$ forms a $2$-path. Thus, $v$ can collect these vectors to determine the total number of $2$-paths from $v$ to $p$. 
This lets us calculate a linear combination involving cycle subgraph counts with an equation that is linearly independent from the others in our system. 

\noindent\textbf{Gather:} At each $v$, take a union of histograms from each neighbor $a$, resolving duplicate entries with the reduce operation 
$(p,c_{a_1} ) \oplus (p,c_{a_2}) = (p, c_{a_1} + c_{a_2})$.

\noindent\textbf{Apply:} Given the gathered histogram vector $\{\mathop{\oplus}_{a \in \Gamma(v)} c_a[p] \}_{p \notin \Gamma(v)}$, calculate the number of non-induced $4$-cycles involving $p$ and two neighbors:
\begin{align}
\sum_{p \notin \Gamma(v)} \binom{\mathop{\oplus}_{a \in \Gamma(v)} c_a[p]}{2} &= F_7(v) + F_9(v) . \label{eq:nbhPivot}
\end{align}

Next, we upper bound savings from our $2$-hop histogram by analyzing the improvement when the only information transmitted across the network to a vertex $v$ is each non-neighboring vertex and its final count $\oplus_{a \in \Gamma(v)} c[p]$. Let $h_v~=~|\Gamma(\Gamma(v)) \; \backslash \; \{\Gamma(v)  \cup v \}|$. For each $v$, the difference between full and histogram information is at most $\sum_{a \in \Gamma(v)} (|\Gamma(a)| - 1)  - 2h_v$.
The exact benefit of \eqref{eq:nbhPivot} depends on the internal implementation of the reduce operation $\oplus$ as pairs of neighbors are gathered. 

Counting the number of distinct pairs of $2$-paths to each $2$-hop neighbor, \textit{i.e.} $\frac{1}{2}(c[p]^2 - c[p])$, requires counting the second moment of $c$ taken over $h_v$ terms. Due to a result by Alon (\cite{alonMemory2002}, Proposition 3.7), the memory required to count this value exactly (moreover, to approximate it deterministically) is $\Omega(h_v)$. Thus, up to implementation details, our memory use is optimal.

\subsection{Normalization and symmetry}
Our final local equation comes from summing the local $4$-profile across all $17$ automorphisms: 
\begin{align}
\sum_{i}^{17} F_i(v) = \binom{|V| - 1}{3} .
\end{align}

To calculate the global $4$-profile, we utilize global symmetry and scaling equations. Let $F_i = \sum_{v \in V}F_i(v)$. Globally, each subgraph count is in exact proportion with the same subgraph counted from a different vertex automorphism. The ratio depends on the subgraph's degree distribution:
\begin{align}
\begin{split}
F_3 &= 2F_3^{'} , \quad F_4 = F_4^{'} , \quad F_6 = 3F_6^{'} , \\
F_8 &= F_8^{'} , \quad F_8^{''} = 2F_8 , \quad F_9 = F_9^{'} . 
\end{split}
\end{align}

Global symmetry makes the equation for $F_8$ and the system \eqref{eq:local3end} linearly dependent. We sum across vertices, inverting a single $11 \times 11$ system to yield the final global $4$-profile $[N_0 , \ldots N_{10}^T]$ by scaling appropriately:
\begin{align}
\begin{split}
N_0 & = \frac{F_0}{4} , \quad N_1 = \frac{F_1}{2} , \quad N_2 = \frac{F_2}{4} , \quad N_3 = F_3^{'} , \\
N_4 &= \frac{F_4^{'}}{2} , \quad N_5 = \frac{F_5}{3}, \quad N_6 = F_6^{'} , \quad N_7 = \frac{F_7}{4}, \\
N_8 &= F_8 , \quad N_9  = \frac{F_9}{2} , \quad N_{10} = \frac{F_{10}}{4} .
\end{split}
\end{align}

\vspace{-2mm}
\begin{algorithm}[ht]
    \caption{\prof}
   \label{alg:4local}
\begin{algorithmic}[1]
    \STATE {\bfseries Input:} Graph $G(V,E)$ with $|V|$ vertices, $|E|$ edges. 
    \STATE {\bfseries Gather:} For each vertex $v$ union over edges of the `other' vertex in the edge, $\cup_{a \in \Gamma(v)} a = \Gamma(v)$.
    \STATE {\bfseries Apply:} Store the gather as vertex data $\texttt{v.nb}$, size automatically stored.
    \STATE {\bfseries Scatter:} For each edge $e_{va}$, compute and store scalars in \eqref{eq:local3}.
    \STATE {\bfseries Gather:} For each edge $e_{va}$, sum edge scalar data of neighbors in \eqref{eq:global4pivot} - \eqref{eq:local3end} and combine two-hop histograms.
    \STATE {\bfseries Apply:} For each vertex $v$, sum over $p \notin \Gamma(v)$ in \eqref{eq:nbhPivot}, store other data in array \texttt{v.u}. No Scatter.
    \STATE {\bfseries Gather:} For each vertex $v$ collect pairs of connected neighbors in $\Delta(v)$.
    \STATE {\bfseries Apply:} Store connected neighbor (triangle) list as vertex data $\texttt{v.conn}$. No Scatter.
    \STATE {\bfseries Gather:} For each vertex $v$ sum \eqref{eq:h8direct}.
    \STATE {\bfseries Apply:} Append data to array $\texttt{v.u}$. Multiply \texttt{v.u} by a matrix to solve system of equations.
    \RETURN [\texttt{v: v.N0 v.N1 v.N2~\ldots~v.N10}]
\end{algorithmic}
\end{algorithm}

%% file: Sampling.tex
In this section, we describe the process for approximating the exact number of subgraphs in a graph $G$. Denote the exact counts by $[N_0 \ldots N_{10}]^T$ and the estimates by $[X_0 \ldots X_{10}]^T$.

We are sparsifying the original graph $G$ by keeping each edge independently with probability $p$.
Denote the random subsampled graph by $\tilde{G}$ and its global $4$-profile by $[Y_0 \ldots Y_{10}]^T$. Clearly each triangle survives with probability $p^3$  and each $4$-clique survives with $p^6$. Therefore, in expectation, $\mathbb{E}[Y_{10}] = p^6 N_{10}$ and $X_{10} = \frac{1}{p^6}Y_{10}$ is unbiased. 

This simple correspondence does not hold for other subgraphs: each triangle in $\tilde{G}$ can only be a triangle in $G$ that survived edge removals, but other subgraphs of $\tilde{G}$ could be originating from multiple subgraphs of $G$ depending on the random sparsification process. 
We can, however, relate the original $4$-profile vector to the expected subsampled $4$-profile vector by a matrix multiplication. 
Let $F(abcd)$ and $\tilde{F}(abcd)$ represent the induced $4$-subgraph on the vertices $abcd$ before and after subsampling, respectively. Then define $\mathbf{H}$ by $H_{ij} = \mathbb{P}(\tilde{F}(abcd) = F_i \:|\: F(abcd) = F_j)$. Thus, we form an unbiased estimator, \textit{i.e.} $\mathbb{E}[X_{i}] = N_{i},~i=1, \ldots, 10$, by inverting the edge sampling matrix.

For $3$-profiles, this process is described by the following system of equations:
\begin{align}
  \begin{bmatrix}
    \mathbb{E}[Y_0] \\ \mathbb{E}[Y_1] \\ \mathbb{E}[Y_2] \\ \mathbb{E}[Y_3]
   \end{bmatrix}
    = \begin{bmatrix}
	  1 & 1-p & (1-p)^2 & (1-p)^3 \\
	  0 & p & 2p(1-p) & 3p(1-p)^2 \\
	  0 & 0 & p^2 & 3p^2(1-p) \\
	  0 & 0 & 0 & p^3
	\end{bmatrix}
	\begin{bmatrix}
		n_0 \\ n_1 \\ n_2 \\ n_3
	\end{bmatrix} .
\end{align}

For $4$-profiles, the vectors are $11$ dimensional and a similar linear system can be explicitly computed -- we include the equations in the Appendix. This matrix turns out to be invertible and we can therefore calculate the $4$-profile exactly if we have access to the expected values of the sparsified $4$-profile. Of course, we can only obtain one sample random graph and calculate that $4$-profile, which will be an accurate estimate if the $4$-profile quantities are sufficiently concentrated around their expectation.

\subsection{Graph k-profile concentration} 
Previous work used this idea of graph sparsification for triangle counting~\cite{Tsourakakis2011sparsifier} and $3$-profiles~\cite{ourpaperKDD}.
The main concentration tool used was the Kim and Vu polynomial concentration~\cite{KimVu2000concentration,Tsourakakis2011sparsifier} which unfortunately gives very loose bounds for practical graph sizes. Figure \ref{fig:accuracy1ACC} compares the accuracy bound derived in this section to the bound predicted by \cite{KimVu2000concentration}. Clearly the Kim-Vu concentration does not provide meaningful bounds for the experiments in Section \ref{sec:results}. However, our results match observed sparsifier accuracy much more closely.

Our novel concentration results that exploit the fact that partial derivatives of the desired quantities are sparse in the number of edge variables. This allows us to use a novel information theoretic concentration technique called read-$k$ functions~\cite{GavinskyReadK}.
For simplicity, we only explain the concentration of $4$-cliques ($F_{10}$ subgraphs) here. 
We establish the general result for all $11$ $4$-profile variables in the Appendix. Our main concentration result is as follows:
\begin{theorem}\label{lem:cliqueRK}
Let $G$ be a graph with $N_{10}$ $4$-cliques, and let $k_{10}$ be the maximum number of $4$-cliques sharing a common edge. Let $X_{10}$ be the $4$-clique estimate obtained from subsampling each edge with probability $0 < p \leq 1$, choose $0 < \delta < 1$, and choose $\epsilon_{RK} > 0$. If 
\begin{align*}
p &\geq \left( \frac{\log(2/\delta)k_{10}}{2 \epsilon_{RK}^2 N_{10}} \right)^{1/12},
\end{align*}
then $\left| N_{10} - X_{10} \right| \leq \epsilon_{RK} N_{10}$ with probability at least $1 - \delta$.
\end{theorem}
\begin{proof}
Our proof relies on read-$k$ function families \cite{GavinskyReadK}, a recent characterization of dependencies among functions of random variables. See the Appendix for full details.
\end{proof}

Next, we state conditions under which our method outperforms the Kim and Vu concentration results~\cite{KimVu2000concentration}. Proof can be found in the Appendix:
\begin{cor}\label{KV_RKcomparison}
Let $G$ be a graph with $m$ edges. If $p = \Omega(1/\log m)$ and $\delta = \Omega(1/m)$, then read-$k$ provides better triangle sparsifier accuracy than Kim-Vu. If additionally $k_{10} \leq N_{10}^{5/6}$, then read-$k$ provides better $4$-clique sparsifier accuracy than Kim-Vu.
\end{cor}

We note that the asymptotic condition on $p$ in Corollary \ref{KV_RKcomparison} includes a constant term much less than $1$. This implies our concentration result is superior over all $p$ values of practical interest. While these bounds contain the quantities we wish to estimate, they provide guidelines for the performance of sampling heuristics. We also investigate this in Section \ref{sec:results} for some realistic graphs.

%% file: Experiments.tex
Let us now describe the implementation and experimental results of our algorithm. We implement \prof\ on GraphLab v2.2 (PowerGraph) \cite{powergraphGAS2012} and measure its running time and accuracy on large input graphs.\footnote{Available at \url{http://github.com/eelenberg/4-profiles}} First, we show that edge sampling yields very good approximation results for global 4-profile counts and achieves substantial execution speedups and network traffic savings when multiple machines are in use.
Due to its distributed nature, we can show \prof\ runs substantially faster when using multiple CPU cores and/or machines. Notice that multicore and multiple machines can not speed up some centralized algorithms, \textit{e.g.}, $\orca$ \cite{Hocevar2014bio}, which we use as a baseline for our results. Note also that $\orca$ produces only a partial 4-subgraph count, \textit{i.e.}, it calculates only connected 4-subgraphs, while $\prof$ calculates all $17$ per vertex.\\
\textbf{The systems:}
We perform the experiments on two systems. The first system is a single powerful server, further referred to as Asterix. The server is equipped with 256 GB of RAM and two Intel Xeon E5-2699 v3 CPUs, 18 cores each. Since each core has two hardware threads, up to 72 logical cores are available to the GraphLab engine.
The second system is an EC2 cluster on AWS\footnote{Amazon Web Services - \url{http://aws.amazon.com}}. The cluster is comprised of 20 c3.8xlarge machines, each having 60 GB RAM and 32 virtual CPUs.\\
\textbf{The data:}
In our experiments we use two real graphs representing different datasets: social networks (LiveJournal: 4,846,609 vertices, 42,851,237 edges) and a WWW graph of Notre Dame (WEB-NOTRE: 325,729 vertices, 1,090,108 edges) \cite{snapnets}. 
Notice that the above graphs are originally directed, but since our work deals with undirected graphs, all duplicate edges (\textit{i.e.}, bi-directional) were removed and directionality is ignored.
\begin{figure}[ht]
\centering
\begin{subfigure}{0.49\linewidth}
\includegraphics[width=\linewidth]{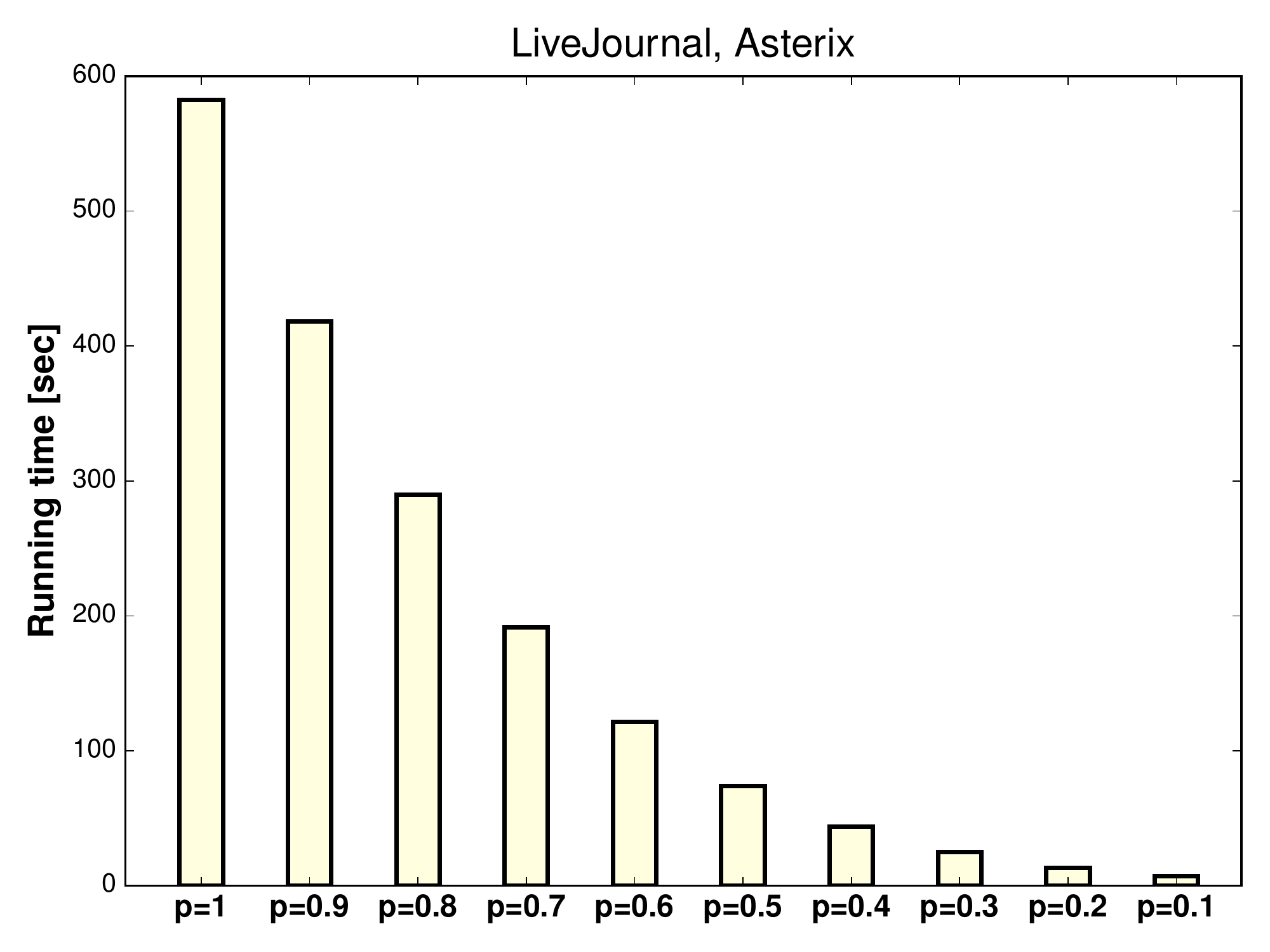}
\caption{}
\label{fig:accuracyTIME}
\end{subfigure}
\begin{subfigure}{0.49\linewidth}
\includegraphics[width=\linewidth]{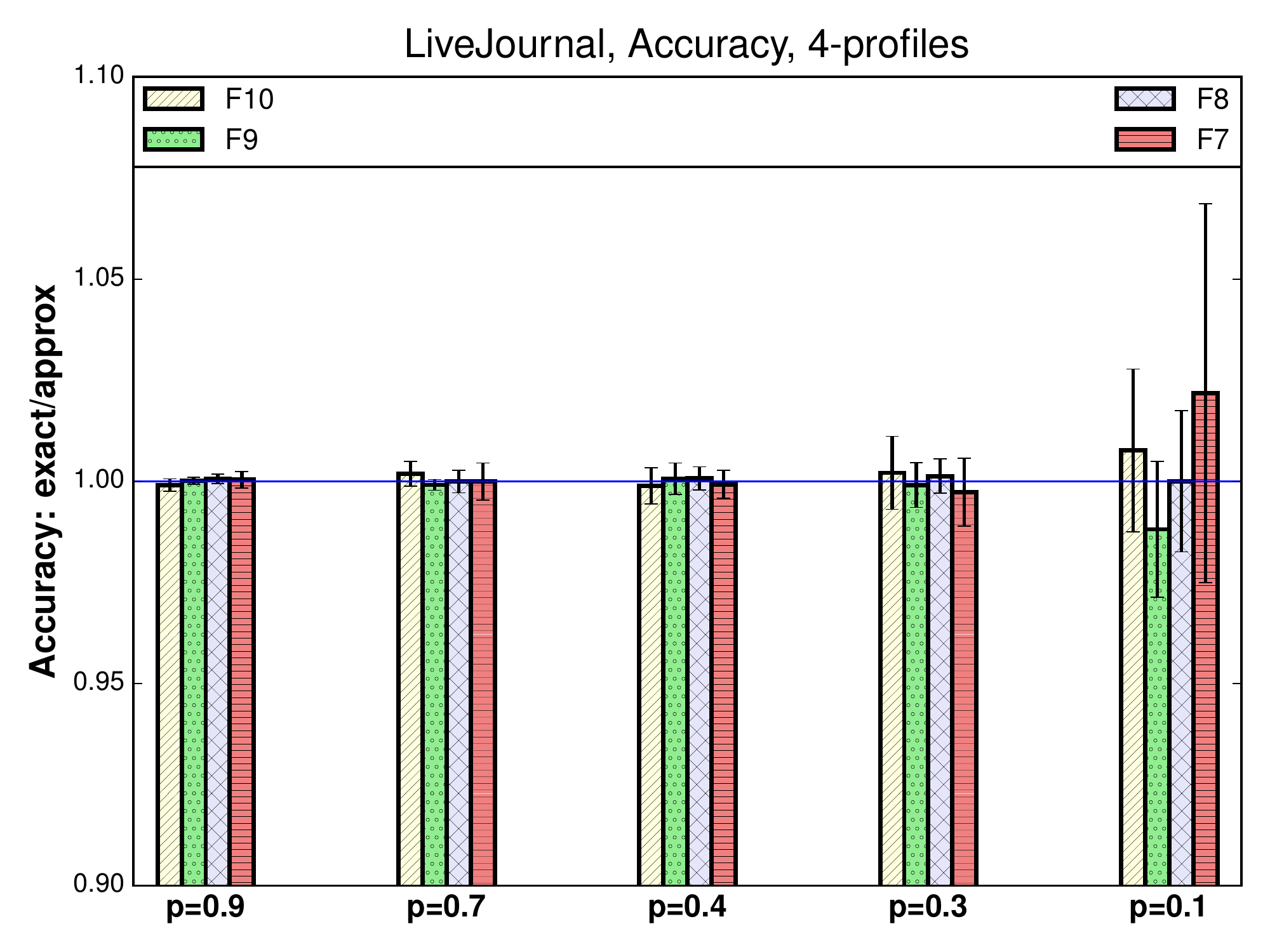}
\caption{}
\label{fig:accuracyBARS}
\end{subfigure}
\vspace{-2mm}
\caption{\label{fig:accuracy} LiveJournal graph, Asterix system. All the results are averaged over $10$ iterations.
(a) -- Running time as a function of sampling probability.
(b) -- Accuracy of the $F_7 - F_{10}$ global counts, measured as ratio of the exact count and the estimated count. 
}
\vspace{-4mm}
\end{figure}

\begin{figure}[h]
\centering
\begin{subfigure}{0.49\linewidth}
\includegraphics[width=\linewidth]{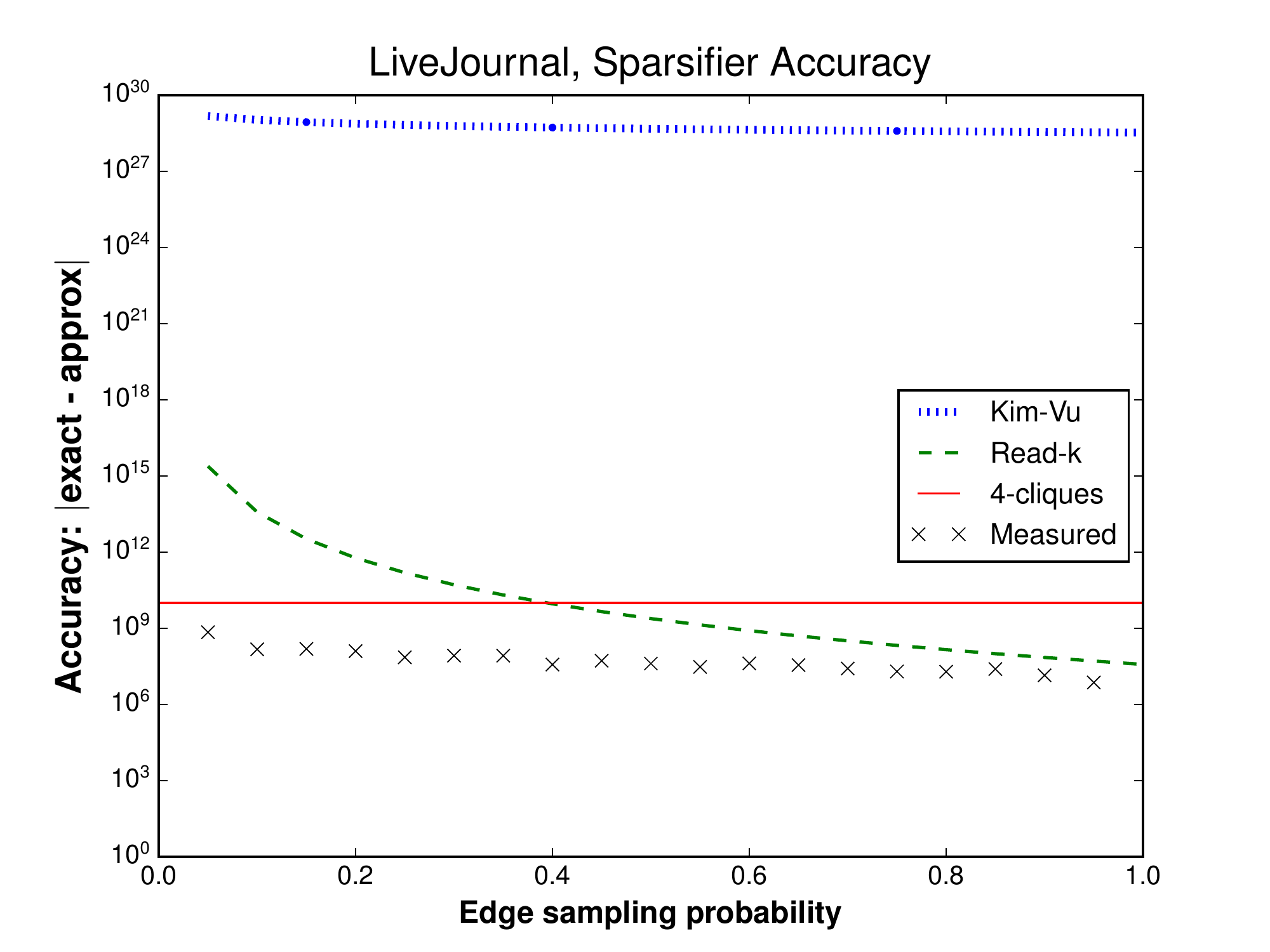}
\caption{}
\label{fig:accuracy1ACC}
\end{subfigure}
\begin{subfigure}{0.49\linewidth}
\includegraphics[width=\linewidth]{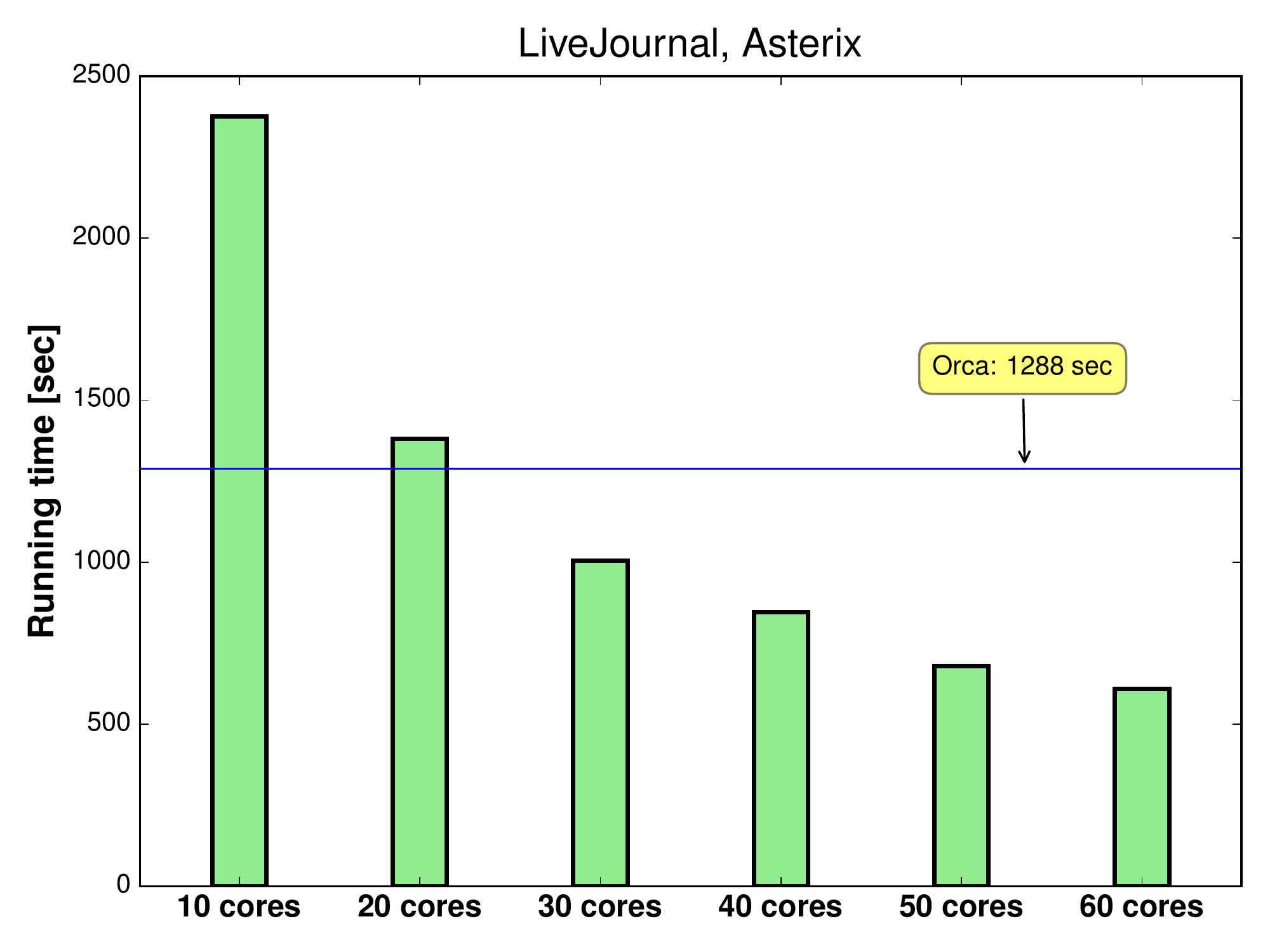}
\caption{}
\label{fig:accuracy1CORES}
\end{subfigure}
\vspace{-2mm}
\caption{\label{fig:accuracy1}LiveJournal graph, Asterix system. All the results are averaged over $10$ iterations.
(a) -- Comparison of $4$-clique sparsifier concentration bounds with accuracy measured in edge sampling experiments. 
(b) -- Comparison of running times of $\orca$ and our exact $\prof$ algorithm. Clearly, $\prof$ benefits from the use of multiple cores.
}
\vspace{-3mm}
\end{figure}

\subsection{Results}\label{sec:results}
\noindent\textbf{Accuracy:} 
The first result is that our edge sampling approach greatly improves running time while maintaining a very good approximation of the \emph{global} 4-profile. In Figure \ref{fig:accuracyTIME} we can see that the running time decreases drastically when the sampling probability decreases. At the same time, Figure \ref{fig:accuracyBARS} shows that the mean ratio of true to estimated global $4$-profiles is within $\pm 2.5\%$. Similar to \cite{Jha2014}, which uses a more complex sampling scheme to count connected $4$-subgraphs, this ratio is usually much less than 1\%. We show here only profiles $F_7 - F_{10}$ since their counts are the smallest and were observed to have the lowest accuracy. In Figure \ref{fig:accuracy1ACC} we compare theoretical concentration bounds on a logarithmic scale and show the benefit of Theorem \ref{lem:cliqueRK}. While the guarantees provided by Kim-Vu \cite{KimVu2000concentration} bounds are very loose (the additive error is bounded by numbers which are orders of magnitude larger than the true value), the read-$k$ approach is much closer to the measured values. We can see that for large sampling probabilities ($p \geq 0.5$), the measured error at most $2$ orders of magnitude smaller than the value prediced by Theorem \ref{lem:cliqueRK}. \\
\noindent\textbf{2-hop histogram:} Now we compare two methods of calculating the left hand side of \eqref{eq:nbhPivot} from Section \ref{subsec:nh}.
We show that a simple implementation in which a vertex gathers its full 2-hop neighborhood (\textit{i.e.}, IDs of its neighbors' neighbors) is much less efficient than the \emph{two-hop histogram} approach used in $\prof$ (see Section \ref{subsec:nh}). 
In Figures \ref{fig:2hopTimeAWS} and \ref{fig:2hopNetAWS} we can see that the histogram approach is an order of magnitude faster for various numbers of machines, and that its network requirements are up to $5$x less than that of the simple implementation. Moreover, our algorithm could handle much larger graphs while the simple implementation ran out of memory.\\
\noindent\textbf{Running time:} 
Finally, we show that $\prof$ can run much faster than the current state of the art graphlet counting implementations. The algorithm and the GraphLab platform on which it runs are both distributed in nature. The latter allows $\prof$ to exploit multiple cores on a single machine as well as a cluster of machines. Figure \ref{fig:accuracy1CORES} shows running time as a function of CPU cores. We compare this result to the running time of a single core, C++ implementation of $\orca$ \cite{Hocevar2014bio}. Our $\prof$ algorithm becomes faster after only $25$ cores and is $2$x faster using $60$ cores. Moreover, $\prof$ allows scaling to a large number of machines. In Figure \ref{fig:2hopTimeLJ} we can see how the running time for the LiveJournal graph decreases when the number of machines increases. Since $\orca$ cannot benefit from multiple machines, we see that $\prof$ runs up to $12$x faster than $\orca$. This gap widens as the cluster grows larger. In \cite{MilinkovicGPU2015}, the authors implemented a GPU version of $\orca$ using CUDA. However, the reported speedup is about $2$x which is much less than we show here on the AWS cluster (see Figure \ref{fig:2hopTimeLJ} for $p=1$).
We also note a substantial running time benefit of the sampling approach for \emph{global} 4-profiles. In Figures \ref{fig:2hopTimeLJ} and \ref{fig:2hopNetLJ}, we see that with $p=0.1$ we can achieve order of magnitude improvements in both speed and network traffic. This sampling probability maintains very good accuracy, as shown in Figure \ref{fig:accuracyBARS}.

\begin{figure}[h]
\centering
\begin{subfigure}{0.49\linewidth}
\includegraphics[width=\linewidth]{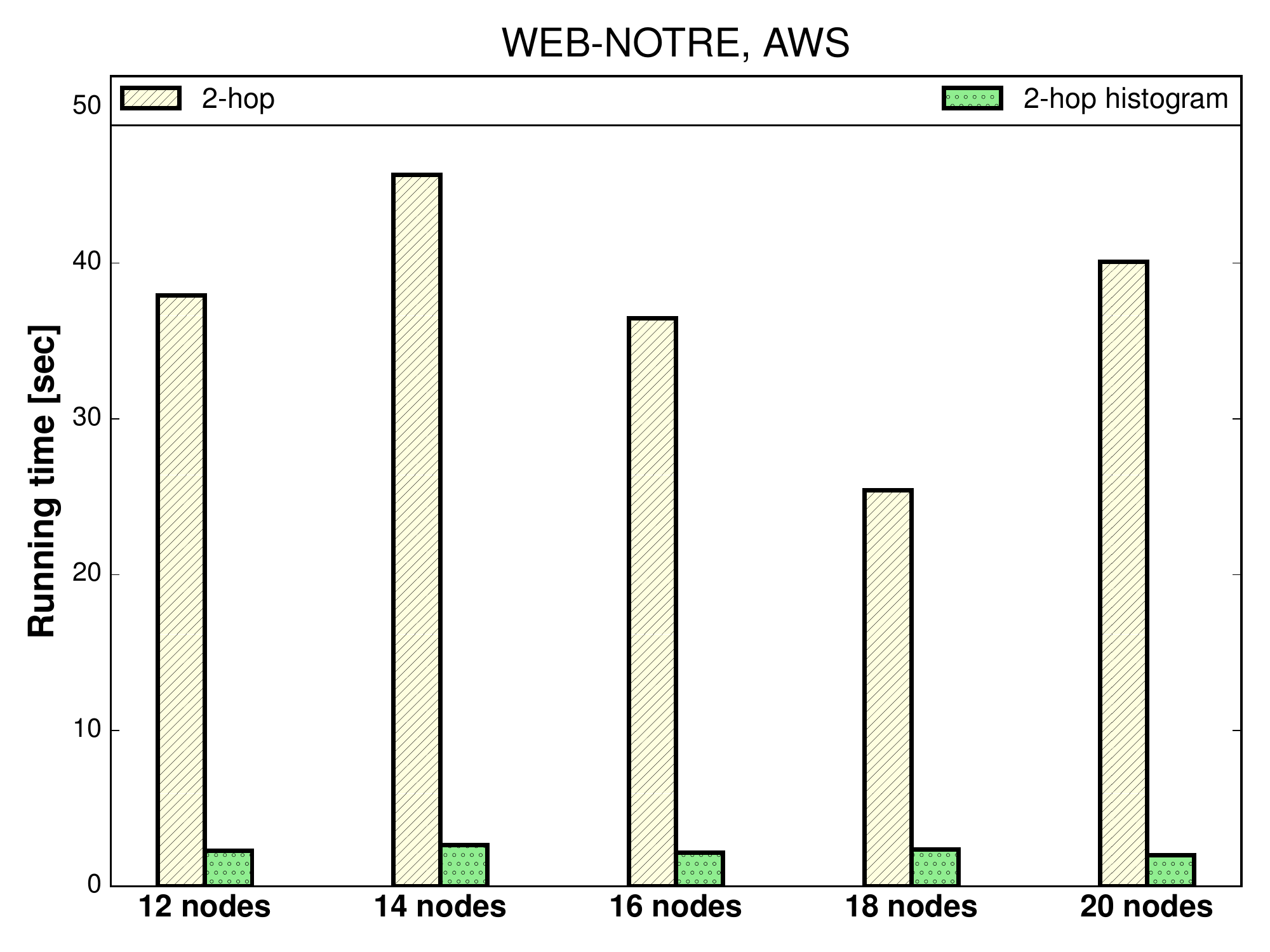}
\caption{}
\label{fig:2hopTimeAWS}
\end{subfigure}
\begin{subfigure}{0.49\linewidth}
\includegraphics[width=\linewidth]{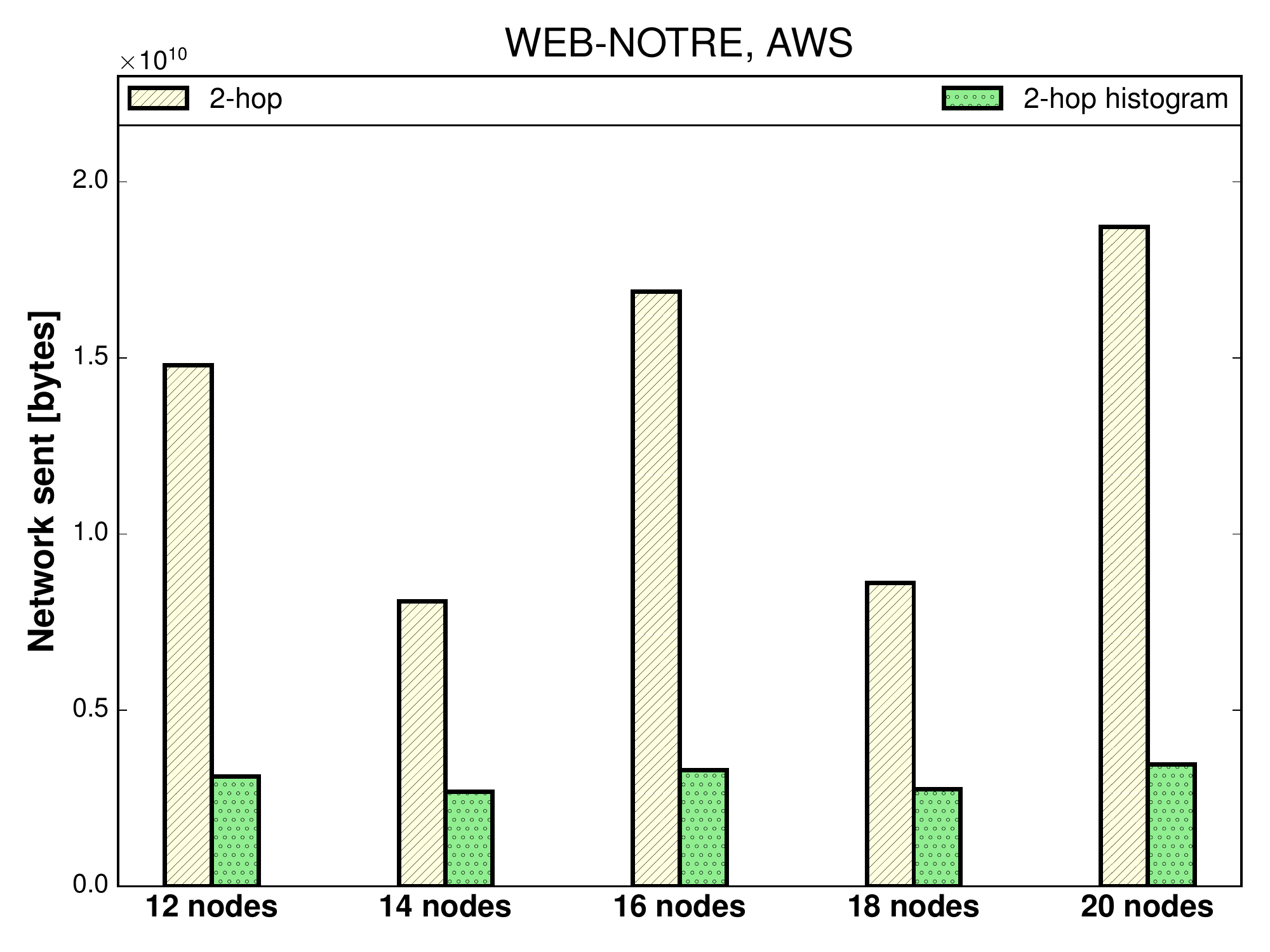}
\caption{}
\label{fig:2hopNetAWS}
\end{subfigure} \\
\begin{subfigure}{0.49\linewidth}
\includegraphics[width=\linewidth]{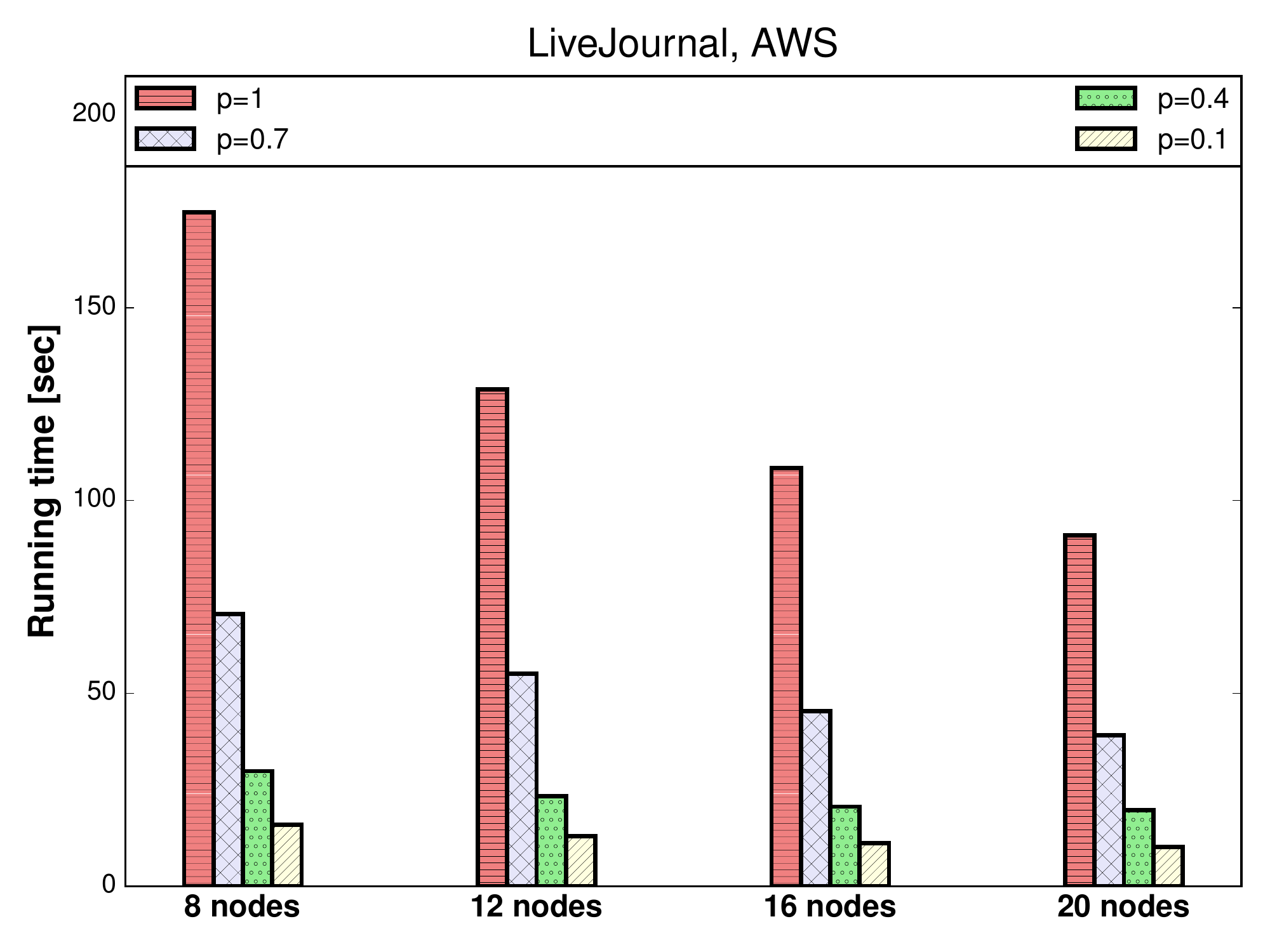}
\caption{}
\label{fig:2hopTimeLJ}
\end{subfigure}
\begin{subfigure}{0.49\linewidth}
\includegraphics[width=\linewidth]{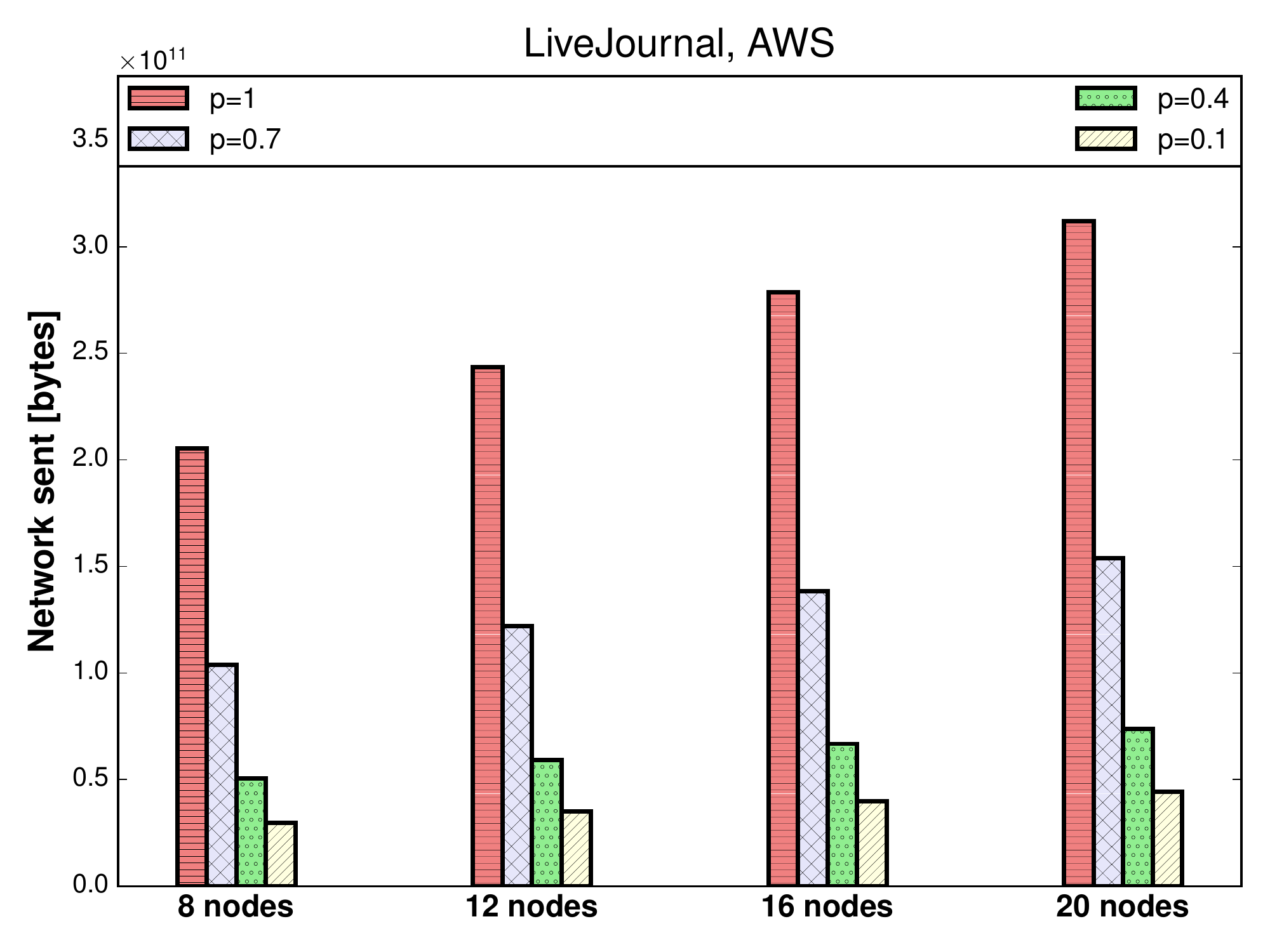}
\caption{}
\label{fig:2hopNetLJ}
\end{subfigure}
\caption{\label{fig:2hopandmore}AWS cluster of up to 20 machines (nodes). 
(a,b) -- Running time and network usage comparing naive $2$-hop implementation and $2$-hop histogram approach on the Notre Dame web graph. 
(c,d) -- Running time and network usage of $\prof$ for various number of compute nodes and sampling probability $p$, on the LiveJournal graph.
All results are averaged over $10$ iterations.}
\end{figure}

\section{Conclusions}
 We introduced a novel distributed algorithm for estimating $4$-profiles of large graphs. We relied on two theoretical results that can be of independent interest: that $4$-profiles can be estimated with limited $2$-hop information and that randomly erasing edges gives sharper approximation compared to previous analysis. We showed that our scheme outperforms the previous state of the art and can exploit cloud infrastructure to scale.

%% file: Appendix.tex
\subsection{Proof of Theorem \ref{lem:cliqueRK}}
\begin{proof}
Rather than Lipschitz bounding the value of each partial derivative, as in \cite{KimVu2000concentration,Tsourakakis2011sparsifier,ourpaperKDD}, our main technical tool \cite{GavinskyReadK} benefits from the fact that each first partial derivative is sparse in the number of edge variables:

\begin{defin}[Read-$k$ Families]
Let $X_1, \ldots, X_m$ be independent random variables. For $j \in [r]$, let $P_j \subseteq [m]$ and let $f_j$ be a Boolean function of $\{X_i\}_{i \in P_j}$. Assume that $|\{j | i \in P_j\}| \leq k$ for every $i \in [m]$. Then the random variables $Z_j = f_j(\{X_i\}_{i \in P_j})$ are called a read-$k$ family.
\end{defin}

Each variable only affects $k$ of the $r$ Boolean functions. Let $G$ be a graph with $N_{10}$ $4$-cliques and a maximum of $k_{10}$ $4$-cliques sharing a common edge. The corresponding $4$-clique estimator $X_{10}$ follows this exact structure. Each edge sampling variable appears in at most $k_{10}$ of the $N_{10}$ terms. We now state the main result required for our analysis. Note that when applied to estimating the number of $4$-cliques, the bound is a function of $k_{10}$ and $N_{10}$ independent of the number of edges. Therefore, it is much stronger than arguments involving Lipschitz bounded functions such as McDiarmid's inequality.

\begin{prop}[Concentration of Read-$k$ Sums \cite{GavinskyReadK}]\label{thm:readk}
Let $Z_1, \ldots, Z_r$ be a family of read-$k$ indicator variables with $\mathbb{P}(Z_i = 1) = p_i$, and let $p$ be the average of $p_1, \ldots , p_r$. Then for any $\epsilon > 0$,
\begin{align*}
\mathbb{P}\left(\sum_{i=1}^r Z_i \geq (p + \epsilon)r \right) \leq \exp \left(- D(p + \epsilon \parallel p) \frac{r}{k}\right) \leq \exp \left( -2 \epsilon^2 \frac{r}{k} \right) \\
\mathbb{P}\left(\sum_{i=1}^r Z_i \leq (p - \epsilon)r \right) \leq  \exp \left(- D(p - \epsilon \parallel p) \frac{r}{k} \right) \leq \exp \left( -2 \epsilon^2 \frac{r}{k} \right) ,
\end{align*}
where $D(x\parallel y) = x \log\left(\frac{x}{y}\right) + (1-x)\log\left(\frac{1-x}{1-y}\right)$ is the Kullback-Leibler divergence of $x$ and $y$.
\end{prop}

Let $Y_{10} = \sum_{\boxtimes(a,b,c,d) \in {\cal H}_{10}} t_{ab}t_{bc} t_{cd} t_{da} t_{ac} t_{bd}$. Then
\begin{align*}
\mathbb{P} \left( \lvert Y_{10} - p^6N_{10} \rvert \geq \epsilon_{RK} N_{10} \right) &\leq 2 \exp \left( -\frac{2 \epsilon_{RK}^2 N_{10}}{k_{10}} \right) \\
\Rightarrow \mathbb{P} \left( \lvert X_{10} - N_{10} \rvert \geq \epsilon_{RK} N_{10} \right) &= \mathbb{P} \left( \lvert Y_{10} - p^6N_{10} \rvert \geq p^6\epsilon_{RK} N_{10} \right) \leq 2 \exp \left( - \frac{2 p^{12} \epsilon_{RK}^2 N_{10}}{k_{10}} \right).
\end{align*}
The claim follows by setting the right hand side less than $\delta$ and solving for $p$.
\end{proof}

\subsection{Proof of Corollary \ref{KV_RKcomparison}}
\begin{proof}
We prove this result for the case of $4$-cliques only because the case for triangles is similar. First we must derive a similar $4$-clique concentration bound using the techniques in \cite{KimVu2000concentration,Tsourakakis2011sparsifier}.

\begin{lem}\label{lem:cliqueKV}
Let $G$ be a graph with $m$ edges and $N_{10}$ cliques, and $k_{10}$ be the maximum number of $4$-cliques sharing a common edge. Let $a_6=8^6\sqrt{6!}$, $0 < p \leq 1$, and $\epsilon_{KV} > 0$. Let $X_{10}$ be the $4$-clique estimate obtained from subsampling each edge with probability $p$. If
\begin{align}
\frac{p}{\max \left\{ \sqrt[6]{1/N_{10}},\sqrt[3]{k_{10}/N_{10}} \right\} } \geq \frac{a_6^2\log^{12}(m^{5+\gamma})}{\epsilon_{KV}^2}, \label{eq:cliqueKVcond}
\end{align}
then $\left| N_{10} -  X_{10} \right| \leq \epsilon_{KV} N_{10}$ with probability at least $1 - \frac{1}{m^\gamma}$.
\end{lem}
\begin{proof}
This proof is a straightforward application of the main result in \cite{KimVu2000concentration}, repeated below for completeness. 
Let $\vecalpha = (\alpha_1, \alpha_2, \ldots, \alpha_m) \in \mathbb{Z}_+^m$ and define $\mathbb{E}_{\geq 1}[X] = \max_{\vecalpha: \pnorm{\vecalpha}{1} \geq 1} \mathbb{E}(\partial^\vecalpha X)$, where
\begin{align}
\mathbb{E}(\partial^\vecalpha X) = \mathbb{E} \left[(\frac{\partial}{\partial t_1})^{\alpha_1} \ldots (\frac{\partial}{\partial t_m})^{\alpha_m} \left[ X(t_1,\ldots, t_m) \right] \right].
\end{align}
Further, we call a polynomial totally positive if the coefficients of all the monomials involved are non-negative.

\begin{prop}[Kim-Vu concentration of multivariate polynomials \cite{KimVu2000concentration}] \label{KIMVUCONCENTRATION}
Let $Y$ be a random totally positive Boolean polynomial in $m$ Boolean random variables with degree at most $k$. If $\mathbb{E}[Y] \geq \mathbb{E}_{\geq 1}[Y]$, then
\begin{align}
\mathbb{P}\left( |Y - \mathbb{E}[Y] | > a_k \sqrt{\mathbb{E}[Y] \mathbb{E}_{\geq 1}[Y]}  \lambda^k \right) = \mathcal{O}\left(\exp \left(-\lambda + (k-1) \log m \right) \right)  \label{eq:kvtheorem}
\end{align}
 for any $\lambda > 1$, where $a_k = 8^k k!^{1/2}$.
\end{prop}

Let $Y_{10} = \sum_{\boxtimes(a,b,c,d) \in {\cal H}_{10}} t_{ab}t_{bc} t_{cd} t_{da} t_{ac} t_{bd}$. Clearly $Y_{10}$ is totally positive. Let $k_{10,ab}$, $\sigma_{abc}$, and $\nu_{abc}$ be the maximum number of $4$-cliques sharing a common edge $t_{ab}$, wedge $\Lambda_{abc}$, and triangle $\Delta_{abc}$, respectively. Taking repeated partial derivatives,
\begin{align*}
    \mathbb{E} \left[ \frac{\partial Y_{10}}{\partial t_{ab}} \right] &= p^5 k_{10,ab} , \\
     \mathbb{E} \left[\frac {\partial Y_{10}}{\partial t_{ab} t_{bc}} \right]  &= p^4 \sigma_{abc} , \quad \mathbb{E} \left[\frac {\partial Y_{10}}{\partial t_{ab} t_{cd}} \right]  = p^4 , \\
     \mathbb{E} \left[\frac {\partial Y_{10}}{\partial t_{ab} t_{bc} t_{ac}} \right]  &= p^3 \nu_{abc} , \quad \mathbb{E} \left[\frac {\partial Y_{10}}{\partial t_{ab} t_{bc} t_{cd}} \right]  = p^3 , \\
     \mathbb{E} \left[\frac{\partial Y_{10}}{\partial  t_{ab} t_{bc} t_{ac} t_{da}} \right]  &=  \mathbb{E} \left[\frac{\partial Y_{10}}{\partial  t_{ab} t_{bc} t_{cd} t_{da}} \right] = p^2, \\
          \mathbb{E} \left[\frac{\partial Y_{10}}{\partial t_{ab}t_{bc} t_{cd} t_{da} t_{ac}} \right]  &= p , \quad
               \mathbb{E} \left[\frac{\partial Y_{10}}{\partial t_{ab}t_{bc} t_{cd} t_{da} t_{ac} t_{bd}} \right]  = 1 .
\end{align*}

Noting that $\sigma_{abc} \leq \min\{k_{10,ab}, k_{10,bc}\} \leq k_{10}$, similarly $\nu_{abc} \leq k_{10}$, and $p^5 \leq p^4 \ldots \leq 1$, we have $\mathbb{E}_{\geq 1} \left[ Y_1 \right] \leq \max \{1, p^3 k_{10}\}$.
  $\mathbb{E}_{\geq1}\left[Y_{10} \right] \leq \mathbb{E}[Y_{10}] = p^6N_{10}$ implies
\begin{align}
 p \geq \max \{\sqrt[6]{1/N_{10}},  \sqrt[3]{k_{10}/N_{10}} \} . \label{eq:kvcondition}
\end{align}

Choose $\epsilon_{KV} \geq 0$ and let $\epsilon_{KV} \mathbb{E}[Y_{10}] = a_6 \sqrt{\mathbb{E}[Y_{10}] \mathbb{E}_{\geq 1}[Y_{10}]}  \lambda^6$. Applying Proposition \ref{KIMVUCONCENTRATION} to $Y_{10}$ given \eqref{eq:cliqueKVcond} and \eqref{eq:kvcondition}, the right hand side of \eqref{eq:kvtheorem} is $\mathcal{O}(\exp(- \gamma \log m)) = \mathcal{O}(1/m^{\gamma})$. Therefore, the error of the $4$-clique estimator $X_{10}$ is
\begin{align*}
\delta X_{10} = \frac{1}{p^6} \delta Y_{10} = \frac{1}{p^6} (\epsilon_{KV} p^6 N_{10}) = \epsilon N_{10}
\end{align*}
with probability greater than $1 - \frac{1}{m^\gamma}$.
\end{proof}

Now we are ready to prove the corollary by comparing Theorem \ref{lem:cliqueRK} and Lemma \ref{lem:cliqueKV}. Fix $p,\delta,>0$ and $\gamma>1$ such that $p = \Omega(1/\log m)$ and $\delta = m^{-\gamma} = \Omega(1/m)$. Now we analyze the bounds $\epsilon_{KV}$ and $\epsilon_{RK}$.
For any graph and $a_6$ defined in Lemma \ref{lem:cliqueKV},
\begin{align}
\frac{1}{a_6^2} \leq 1, \quad \frac{\gamma}{(5+\gamma)^{12}} \leq 1, \quad \log(2^{1/\gamma}m) \leq 2 \log m, \quad \left(\frac{k_{10}}{N_{10}}\right)^{2/3} \leq 1. \label{eq:comparegeq0}
\end{align}
We further require $k_{10} \leq N_{10}^{5/6}$. Then the condition on $p$ with \eqref{eq:comparegeq0} implies
\begin{align*}
p^{11} \geq 1/\log^{11}(m) 
&\geq \frac{\gamma \log(2^{1/\gamma}m)}{2a_6^2(5+\gamma)^{12}\log^{12}(m)} \min \left\{ k_{10}/N_{10}^{5/6}, (k_{10}/N_{10})^{2/3} \right\}.
\end{align*}
Rearranging terms,
\begin{align*}
\frac{a_6^2\log^{12}(m^{5+\gamma}) \max \left\{ \sqrt[6]{1/N_{10}},\sqrt[3]{k_{10}/N_{10}}\right\}}{p} &\geq \frac{\log(2m^\gamma)k_{10}/N_{10}}{2 p^{12}} \Rightarrow \epsilon_{KV}^2 \geq \epsilon_{RK}^2.
\end{align*}
\end{proof}

We note that the asymptotic condition $p = \Omega(1/\log m)$, includes a constant much smaller than $1$. This is due to the looseness of inequalities in \eqref{eq:comparegeq0} and implies that Theorem \ref{lem:cliqueRK} is superior to Lemma \ref{lem:cliqueKV} over all $p$ values of practical interest.

\subsection{Implementation details}
To improve the practical performance of \prof\  (see Algorithm \ref{alg:4local} for pseudocode), we handle low and high degree vertices differently. As in GraphLab PowerGraph's standard triangle counting, cuckoo hash tables are used if the vertex degree is above a threshold. Now, we also threshold vertices to determine whether the $2$-hop histogram in Section \ref{subsec:nh} will be either a vector or an unordered map. This is because sorting and merging operations on a vector scale poorly with increasing degree size, while an unordered map has constant lookup time. We found that this approach successfully trades off processing time and memory consumption.

\subsection{Extension to global 4-profile sparsifier}

Another advantage to read-$k$ function families is that they are simpler to extend to more complex subgraphs. We now state concentration results for the full $4$-profile sparsifier evaluated experimentally in Section \ref{sec:experiments}

Using the notation in Section \ref{sec:spars}, the edge sampling matrix $\mathbf{H}$ is defined by the relations
\begin{gather*}
\begin{bmatrix}
\mathbb{E}[Y_0] \\ \vdots \\ \mathbb{E}[Y_{10}]
\end{bmatrix}
= \mathbf{H}
\begin{bmatrix}
N_0 \\  \vdots \\ N_{10}
\end{bmatrix} \quad \Rightarrow \quad
\begin{bmatrix}
X_0 \\ \vdots \\ X_{10}
\end{bmatrix}
= \mathbf{H}^{-1}
\begin{bmatrix}
Y_0 \\  \vdots \\ Y_{10}
\end{bmatrix},
\qquad \text{where} \\ \\ \\ \\ \\
\arraycolsep=2pt
\small
\mathbf{H} = 
\begin{bmatrix}
1 & 1-p & (1-p)^2 & (1-p)^2  & (1-p)^3  & (1-p)^3  & (1-p)^3  & (1-p)^4  & (1-p)^4  & (1-p)^5  & (1-p)^6 \\
0 & p & 2p(1-p) & 2p(1-p) & 3p(1-p)^2 & 3p(1-p)^2 & 3p(1-p)^2 & 4p(1-p)^3 & 4p(1-p)^3 & 5p(1-p)^4 & 6p(1-p)^5   \\
0 & 0 & p^2 & 0 & p^2(1-p) & 0 & 0 & 2p^2(1-p)^2 & p^2(1-p)^2 & 2p^2(1-p)^3 & 3p^2(1-p)^4 \\
0 & 0 & 0 &  p^2 & 2p^2(1-p) & 3p^2(1-p) & 3p^2(1-p) & 4p^2(1-p)^2 & 5p^2(1-p)^2 & 8p^2(1-p)^3 & 12p^2(1-p)^4 \\
0 & 0 & 0 & 0 & p^3 & 0 & 0 & 4p^3(1-p) & 2p^3(1-p) & 6p^3(1-p)^2 & 12p^3(1-p)^3 \\ 
0 & 0 & 0 & 0 & 0 & p^3 & 0 & 0 & p^3(1-p) & 2p^3(1-p)^2 & 4p^3(1-p)^3 \\
0 & 0 & 0 & 0 & 0 & 0 & p^3 & 0 & p^3(1-p) & 2p^3(1-p)^2 & 4p^3(1-p)^3 \\
0 & 0 & 0 & 0 & 0 & 0 & 0 & p^4 & 0 & p^4(1-p) & 3p^4(1-p)^2 \\
0 & 0 & 0 & 0 & 0 & 0 & 0 & 0 & p^4 & 4p^4(1-p) & 12p^4(1-p)^2 \\
0 & 0 & 0 & 0 & 0 & 0 & 0 & 0 & 0 & p^5 & 6p^5(1-p) \\
0 & 0 & 0 & 0 & 0 & 0 & 0 & 0 & 0 & 0 & p^6 \\
\end{bmatrix}.
\end{gather*}
\arraycolsep=5pt

Let $t=\frac{p-1}{p} $. Then the inverse sampling matrix is given by
\begin{align}
\mathbf{H}^{-1} = 
\begin{bmatrix}
1 & t & t^2 & t^2  & t^3  & t^3  & t^3  & t^4  & t^4  & t^5  & t^6 \\
0 & \frac{1}{p} & \frac{2t}{p} & \frac{2t}{p} & \frac{3t^2}{p} & \frac{3t^2}{p} & \frac{3t^2}{p} & \frac{4t^3}{p} & \frac{4t^3}{p} & \frac{5t^4}{p} & \frac{6t^5}{p}   \\
0 & 0 & \frac{1}{p^2} & 0 & \frac{t}{p^2} & 0 & 0 & \frac{2t^2}{p^2} & \frac{t^2}{p^2} & \frac{2t^3}{p^2} & \frac{3t^4}{p^2} \\
0 & 0 & 0 &  \frac{1}{p^2} & \frac{2t}{p^2} & \frac{3t}{p^2} & \frac{3t}{p^2} & \frac{4t^2}{p^2} & \frac{5t^2}{p^2} & \frac{8t^3}{p^2} & \frac{12t^4}{p^2} \\
0 & 0 & 0 & 0 & \frac{1}{p^3} & 0 & 0 & \frac{4t}{p^3} & \frac{2t}{p^3} & \frac{6t^2}{p^3} & \frac{12t^3}{p^3} \\ 
0 & 0 & 0 & 0 & 0 & \frac{1}{p^3} & 0 & 0 & \frac{t}{p^3} & \frac{2t^2}{p^3} & \frac{4t^3}{p^3} \\
0 & 0 & 0 & 0 & 0 & 0 & \frac{1}{p^3} & 0 & \frac{t}{p^3} & \frac{2t^2}{p^3} & \frac{4t^3}{p^3} \\
0 & 0 & 0 & 0 & 0 & 0 & 0 & \frac{1}{p^4} & 0 & \frac{t}{p^4} & \frac{3t^2}{p^4} \\
0 & 0 & 0 & 0 & 0 & 0 & 0 & 0 & \frac{1}{p^4} & \frac{4t}{p^4} & \frac{12t^2}{p^4} \\
0 & 0 & 0 & 0 & 0 & 0 & 0 & 0 & 0 & \frac{1}{p^5} & \frac{6t}{p^5} \\
0 & 0 & 0 & 0 & 0 & 0 & 0 & 0 & 0 & 0 & \frac{1}{p^6} \\
\end{bmatrix}.
\label{eq:channelinv4}
\end{align}
The binomial coefficients in these matrices influence our concentration bounds, which we now state:

\input{Appendix-Thm3_E}

%% file: Appendix-Thm3_E.tex
\begin{theorem}[$4$-profile sparsifier estimators]
Consider the sampling process described above and in Section \ref{sec:spars}. Let $X_i,~0 \leq i \leq 10$ (and $\mathbf{X}$ be a vector of these estimates), be the actual estimates of $4$-profiles. Let $k_{i}$ be the maximum number of subgraphs $F_i$ sharing a common edge. Let $Y_i,~0 \leq i \leq 10$, be the $4$ profile counts of the sparsified graph. Then let $N_i,~0 \leq i \leq 10$, be the actual counts. Choose $0 < \delta < 1$ and $\epsilon > 0$.  
Let $C = (192)^2/2$. If 
\begin{align*}
p &\geq \left( \frac{C\log(2/\delta)k_{10}}{ \epsilon^2 N_{10}} \right)^{1/12}, \quad 
p \geq \left( \frac{C\log(2/\delta)(k_9 + 6k_{10})}{\epsilon^2 (N_9 + 6N_{10})} \right)^{1/10}, \quad 
p \geq \left( \frac{C\log(2/\delta)(k_8 + 4k_9 + 12k_{10})}{ \epsilon^2 (N_8 + 4N_9 + 12N_{10})} \right)^{1/8} \\
p &\geq \left( \frac{C\log(2/\delta)(k_7 + k_9 + 3k_{10})}{ \epsilon^2 (N_7 + N_9 + 3N_{10})} \right)^{1/8}, \quad 
p \geq \left( \frac{C\log(2/\delta)(k_6 + k_8 + 2k_9 + 4k_{10})}{ \epsilon^2 (N_6 + N_8 + 2N_9 + 4N_{10})} \right)^{1/6} \\
p &\geq \left( \frac{C\log(2/\delta)(k_5 + k_8 + 2k_9 + 4k_{10})}{\epsilon^2 (N_5 + N_8 + 2N_9 + 4N_{10})} \right)^{1/6}, \quad
p \geq \left( \frac{C\log(2/\delta)(k_4 + 4k_7 + 2k_8 + 6k_9 + 12k_{10})}{ \epsilon^2 (N_4 + 4N_7 + 2N_8 + 6N_9 + 12N_{10})} \right)^{1/6} \\
p &\geq \left( \frac{C\log(2/\delta)(k_3 + 2k_4 + 3k_5 + 3k_6 + 4k_7 + 5k_8 + 8k_9 + 12k_{10})}{ \epsilon^2 (N_3 + 2N_4 + 3N_5 + 3N_6 + 4N_7 + 5N_8 + 8N_9 + 12N_{10})} \right)^{1/4} \\
p &\geq \left( \frac{C\log(2/\delta)(k_2 + k_4 + 2k_7 + k_8 + 2k_9 + 3k_{10})}{ \epsilon^2 (N_2 + N_4 + 2N_7 + N_8 + 2N_9 + 3N_{10})} \right)^{1/4} \\
p &\geq \left( \frac{C\log(2/\delta)(k_1 + 2k_2 + 2k_3 + 3k_4 + 3k_5 + 3k_6 + 4k_7 + 4k_8 + 5k_9 + 6k_{10})}{ \epsilon^2 (N_1 + 2N_2 + 2N_3 + 3N_4 + 3N_5 + 3N_6 + 4N_7 + 4N_8 + 5N_9 + 6N_{10})} \right)^{1/2} \\
n_0 &\leq |V|^2 \left( |V|^2 - \frac{C\log(2/\delta)}{ \epsilon^2} \right),
\end{align*}
then $\pnorm{\delta \mathbf{X}}{\infty} \leq \epsilon \binom{|V|}{4}$ with probability at least $1 - \delta$. 
\end{theorem}
\begin{proof}

We apply Proposition \ref{thm:readk} a total of $11$ times to the sampling-estimator system defined above by $\mathbf{H}$ and $\mathbf{H}^{-1}$. In our context, each sampled subgraph count $Y_i$ is a sum of functions in a  read-$k_{Y_i}$ family, where $k_{Y_i} \leq \min\{|V|-2,N_i\}$. Let $k_{i,e}$ be the maximum number of subgraphs $F_i$ sharing a common edge $e$, and let $k_i= \max_e k_{i,e}$, for $i~=~0,\ldots,10$. The $Y_i$'s have the following parameters:
\begin{align}
\begin{split}
r_{Y_0} &= \binom{|V|}{4}, \quad k_{Y_0} = |V| \\
r_{Y_1} &= N_1 + 2N_2 + 2N_3 + 3N_4 + 3N_5 + 3N_6 + 4N_7 + 4N_8 + 5N_9 + 6N_{10} \\
k_{Y_1} &= k_1 + 2k_2 + 2k_3 + 3k_4 + 3k_5 + 3k_6 + 4k_7 + 4k_8 + 5k_9 + 6k_{10} \\
r_{Y_2} &= N_2 + N_4 + 2N_7 + N_8 + 2N_9 + 3N_{10} \\
k_{Y_2} &= k_2 + k_4 + 2k_7 + k_8 + 2k_9 + 3k_{10} \\
r_{Y_3} &= N_3 + 2N_4 + 3N_5 + 3N_6 + 4N_7 + 5N_8 + 8N_9 + 12N_{10} \\
k_{Y_3} &= k_3 + 2k_4 + 3k_5 + 3k_6 + 4k_7 + 5k_8 + 8k_9 + 12k_{10}\\
r_{Y_4} &= N_4 + 4N_7 + 2N_8 + 6N_9 + 12N_{10}, \quad k_{Y_4} = k_4 + 4k_7 + 2k_8 + 6k_9 + 12k_{10}\\
r_{Y_5} &= N_5 + N_8 + 2N_9 + 4N_{10}, \quad k_{Y_5} = k_5 + k_8 + 2k_9 + 4k_{10} \\
r_{Y_6} &= N_6 + N_8 + 2N_9 + 4N_{10}, \quad k_{Y_6} = k_6 + k_8 + 2k_9 + 4k_{10}\\
r_{Y_7} &= N_7 + N_9 + 3N_{10}, \quad k_{Y_7} = k_7 + k_9 + 3k_{10} \\
r_{Y_8} &= N_8 + 4N_9 + 12N_{10}, \quad k_{Y_8} = k_8 + 4k_9 + 12k_{10} \\
r_{Y_9} &= N_9 + 6N_{10}, \quad k_{Y_9} = k_9 + 6k_{10} \\
r_{Y_{10}} &= N_{10} , \quad k_{Y_{10}} = k_{10}
\end{split}
\end{align}

We show the application of Proposition \ref{thm:readk} for $Y_7$ through $Y_9$ because $Y_{10}$ was shown in the proof of Theorem \ref{lem:cliqueRK} and the other cases are similar:
\begin{align*}
\mathbb{P} \left( \lvert Y_{7} - (p^4N_7 + p^4(1-p) N_9 + 3p^4(1-p)^2N_{10}) \rvert \geq p^4\epsilon (N_7 + N_9 + 3N_{10}) \right) &\leq 2 \exp \left( -\frac{2 p^8 \epsilon^2 (N_7 + N_9 + 3N_{10})}{k_7 + k_9 + 3k_{10}} \right) \\
\mathbb{P} \left( \lvert Y_{8} - (p^4N_8 + 4p^4(1-p) N_9 + 12p^4(1-p)^2N_{10}) \rvert \geq p^4\epsilon (N_8 + 4N_9 + 12N_{10}) \right) &\leq 2 \exp \left( -\frac{2 p^8 \epsilon^2 (N_8 + 4N_9 + 12N_{10})}{k_8 + 4k_9 + 12k_{10}} \right) \\
\mathbb{P} \left( \lvert Y_{9} - (p^5 N_9 + 6p^5(1-p)N_{10}) \rvert \geq \epsilon (N_9 + 6N_{10}) \right) &\leq 2 \exp \left( -\frac{2 \epsilon^2 (N_9 + 6N_{10})}{k_9 + 6k_{10}} \right) \\
\Rightarrow \mathbb{P} \left( \lvert \frac{1}{p^5}Y_9 - (N_{9} + 6(1-p)N_{10} )\rvert \geq \epsilon (N_9 + 6N_{10}) \right) &\leq 2 \exp \left( - \frac{2 p^{10} \epsilon^2 (N_9 + 6N_{10})}{k_9 + 6k_{10}} \right) \\
\mathbb{P} \left( \lvert Y_{10} - p^6N_{10} \rvert \geq \epsilon N_{10} \right) &\leq 2 \exp \left( -\frac{2 \epsilon^2 N_{10}}{k_{10}} \right) \\
\Rightarrow \mathbb{P} \left( \lvert X_{10} - N_{10} \rvert \geq \epsilon N_{10} \right) = \mathbb{P} \left( \lvert Y_3 - p^6N_{10} \rvert \geq p^6 \epsilon N_{10} \right) &\leq 2 \exp \left( - \frac{2 p^{12} \epsilon^2 N_{10}}{k_{10}} \right)
\end{align*}

Rearranging to solve for $p$, we have
\begin{align}
\begin{split}
p \geq \left( \frac{\log(2/\delta)k_{10}}{2 \epsilon^2 N_{10}} \right)^{1/12}, \quad 
p \geq \left( \frac{\log(2/\delta)(k_9 + 6k_{10})}{2 \epsilon^2 (N_9 + 6N_{10})} \right)^{1/10}, \quad 
p \geq \left( \frac{\log(2/\delta)(k_8 + 4k_9 + 12k_{10})}{2 \epsilon^2 (N_8 + 4N_9 + 12N_{10})} \right)^{1/8} \\
p \geq \left( \frac{\log(2/\delta)(k_7 + k_9 + 3k_{10})}{2 \epsilon^2 (N_7 + N_9 + 3N_{10})} \right)^{1/8}, \quad 
p \geq \left( \frac{\log(2/\delta)(k_6 + k_7 + 2k_9 + 4k_{10})}{2 \epsilon^2 (N_6 + N_8 + 2N_9 + 4N_{10})} \right)^{1/6} \\
p \geq \left( \frac{\log(2/\delta)(k_5 + k_7 + 2k_9 + 4k_{10})}{2 \epsilon^2 (N_5 + N_8 + 2N_9 + 4N_{10})} \right)^{1/6}, \quad
p \geq \left( \frac{\log(2/\delta)(k_4 + 4k_7 + 2k_8 + 6k_9 + 12k_{10})}{2 \epsilon^2 (N_4 + 4N_7 + 2N_8 + 6N_9 + 12N_{10})} \right)^{1/6} \\
p \geq \left( \frac{\log(2/\delta)(k_3 + 2k_4 + 3k_5 + 3k_6 + 4k_7 + 5k_8 + 8k_9 + 12k_{10})}{2 \epsilon^2 (N_3 + 2N_4 + 3N_5 + 3N_6 + 4N_7 + 5N_8 + 8N_9 + 12N_{10})} \right)^{1/4} \\
p \geq \left( \frac{\log(2/\delta)(k_2 + k_4 + 2k_7 + k_8 + 2k_9 + 3k_{10})}{2 \epsilon^2 (N_2 + N_4 + 2N_7 + N_8 + 2N_9 + 3N_{10})} \right)^{1/4} \\
p \geq \left( \frac{\log(2/\delta)(k_1 + 2k_2 + 2k_3 + 3k_4 + 3k_5 + 3k_6 + 4k_7 + 4k_8 + 5k_9 + 6k_{10})}{2 \epsilon^2 (N_1 + 2N_2 + 2N_3 + 3N_4 + 3N_5 + 3N_6 + 4N_7 + 4N_8 + 5N_9 + 6N_{10})} \right)^{1/2}
\end{split}
\end{align}

The final condition comes from the result for $Y_0$:
\begin{align}
n_0 &\leq \binom{|V|}{4} - \frac{\log(2/\delta) |V|^2}{2 \epsilon^2} \leq |V|^2 \left( |V|^2 - \frac{\log(2/\delta)}{2 \epsilon^2} \right)
\end{align}

Plugging into our estimators (given by $\mathbf{H}^{-1}$), we get the following error bounds:
\begin{align*}
\delta X_0 &\leq \epsilon(n_1 + n_2 + n_3) + \epsilon (n_1 + 2n_2 + 3n_3 + n_2 + 3n_3 + n_3) \\
& \leq \epsilon (2n_1 + 4n_2 + 8n_3) \leq 8\epsilon \binom{|V|}{3} \\
\delta X_1  &\leq \epsilon(N_1 + 2N_2 + 2N_3 + 3N_4 + 3N_5 + 3N_6 + 4N_7 + 4N_8 + 5N_9 + 6N_{10}) + 2\epsilon(N_2 + N_4 + 2N_7 + N_8 + 2N_9 + 3N_{10}) \\ 
& \quad + 2\epsilon (N_3 + 2N_4 + 3N_5 + 3N_6 + 4N_7 + 5N_8 + 8N_9 + 12N_{10}) + 3 \epsilon (N_4 + 4N_7 + 2N_8 + 6N_9 + 12N_{10}) \\
& \quad + 3 \epsilon (N_5 + N_8 + 2N_9 + 4N_{10}) + 3\epsilon (N_6 + N_8 + 2N_9 + 4N_{10}) + 4\epsilon (N_7 + N_9 + 3N_{10}) \\
& \quad + 4 \epsilon (N_8 + 4N_9 + 12N_{10}) + 5 \epsilon (N_9 + 6N_{10}) + 6\epsilon (N_{10}) \\
& \leq \epsilon (N_1 +  \ldots  + 192N_{10}) \leq 192\epsilon \binom{|V|}{4} \\
\delta X_2  &\leq \epsilon(N_2 + N_4 + 2N_7 + N_8 + 2N_9 + 3N_{10}) + \epsilon (N_4 + 4N_7 + 2N_8 + 6N_9 + 12N_{10}) +2 \epsilon (N_7 + N_9 + 3N_{10}) \\
& \quad + \epsilon (N_8 + 4N_9 + 12N_{10}) + 2 \epsilon (N_9 + 6N_{10}) + 3\epsilon (N_{10}) \\
& \leq \epsilon (N_2 + \ldots + 48N_{10}) \leq 48\epsilon \binom{|V|}{4} \\
\delta X_3  &\leq \epsilon (N_3 + 2N_4 + 3N_5 + 3N_6 + 4N_7 + 5N_8 + 8N_9 + 12N_{10}) + 2 \epsilon (N_4 + 4N_7 + 2N_8 + 6N_9 + 12N_{10}) \\
& \quad + 3 \epsilon (N_5 + N_8 + 2N_9 + 4N_{10}) + 3\epsilon (N_6 + N_8 + 2N_9 + 4N_{10}) + 4\epsilon (N_7 + N_9 + 3N_{10}) + 5 \epsilon (N_8 + 4N_9 + 12N_{10}) \\
& \quad + 8 \epsilon (N_9 + 6N_{10}) + 12\epsilon (N_{10}) \\
& \leq \epsilon (N_3 + 4N_4 + 6N_5 + \ldots + 192N_{10}) \leq 192\epsilon \binom{|V|}{4} \\
\delta X_4  &\leq \epsilon (N_4 + 4N_7 + 2N_8 + 6N_9 + 12N_{10}) + 4\epsilon (N_7 + N_9 + 3N_{10}) + 2 \epsilon (N_8 + 4N_9 + 12N_{10}) + 6 \epsilon (N_9 + 6N_{10}) + 12 \epsilon (N_{10}) \\ 
& \leq \epsilon (N_4 + \ldots + 96N_{10}) \leq 96\epsilon \binom{|V|}{4} \\
\delta X_5  &\leq \epsilon (N_5 + N_8 + 2N_9 + 4N_{10}) + \epsilon (N_8 + 4N_9 + 12N_{10}) + 2\epsilon(N_9 + 6N_{10}) + 4\epsilon (N_{10}) \\
& \leq \epsilon (N_5 + \ldots + 32N_{10}) \leq 32\epsilon \binom{|V|}{4} \\
\delta X_6  &\leq \epsilon (N_6 + N_8 + 2N_9 + 4N_{10}) + \epsilon (N_8 + 4N_9 + 12N_{10}) + 2\epsilon(N_9 + 6N_{10}) + 4\epsilon (N_{10}) \\
& \leq \epsilon (N_6 + \ldots + 32N_{10}) \leq 32\epsilon \binom{|V|}{4} \\
\delta X_7  &\leq \epsilon (N_7 + N_9 + 3N_{10}) + \epsilon(N_9 + 6N_{10}) + 3\epsilon (N_{10}) \\
& \leq \epsilon (N_7 + 2N_9 + 12N_{10}) \leq 12\epsilon \binom{|V|}{4} \\
\delta X_8  &\leq \epsilon (N_8 + 4N_9 + 12N_{10}) + 4 \epsilon(N_9 + 6N_{10}) + 12\epsilon (N_{10}) \\
& \leq \epsilon (N_8 + 8N_9 + 48N_{10}) \leq 48\epsilon \binom{|V|}{4} \\
\delta X_9  &\leq \epsilon (N_9 + 6N_{10}) + 6\epsilon (N_{10}) \\
& \leq \epsilon (N_9 + 12N_{10}) \leq 12\epsilon \binom{|V|}{4} \\
\delta X_{10} &\leq \epsilon N_{10}.
\end{align*}

Thus the maximum deviation in any estimator is less than $192\epsilon \binom{|V|}{4}$. 
Substituting $\tilde{\epsilon}^2 = \epsilon^2/(192)^2 = \epsilon^2/2C$ completes the proof.

\end{proof}